\documentclass[conference]{IEEEtran}

\usepackage{graphicx}
\usepackage{epstopdf}
\usepackage{listings}
\usepackage{xcolor}
\usepackage{amsmath, amsfonts, amssymb}
\usepackage{graphicx,wrapfig,floatflt}

\usepackage{array}
\usepackage{mdwmath}
\usepackage{mdwtab}
\usepackage{eqparbox}
\usepackage{bm,bbm,amssymb,amsthm,url}

\usepackage{algorithm}
\usepackage[noend]{algpseudocode}

\makeatletter
\def\BState{\State\hskip-\ALG@thistlm}
\makeatother

\ifCLASSOPTIONcompsoc \usepackage[caption=false,font=normalsize,labelfont=sf,textfont=sf]{subfig}
\else
\usepackage[caption=false,font=footnotesize]{subfig}
\fi

\numberwithin{equation}{section}

\newcommand{\MAT}{\left[ \begin{array}}  
\newcommand{\mat}{\end{array} \right]}

\newtheorem{Lemma}{Lemma}[section]
\newtheorem{Theorem}{Theorem}[section]

\newtheorem{Q}{Question}[section]
\newtheorem{Proposition}{Proposition}

\def \minimize {\operatorname*{minimize}}

\def \tr {\operatorname*{Tr\ }}
\def \st {\operatorname*{subject\ to\ }}

\def \diag {\operatorname*{diag}}
\def \sign {\operatorname*{sign}}

\def \PPhi {\mathbf{\Phi}}
\def \pphi {\mathbf{\phi}}
\def \0 {\mathbf{0}}
\def \a {\bm{a}}
\def \F {\mathbf{F}}
\def \A {\mathbf{A}}

\def \b {\bm{b}}
\def \B {\mathbf{B}}

\def \c {\bm{c}}
\def \C {\mathbf{C}}

\def \D {\mathbf{D}}

\def \e {\bm{e}}
\def \E {\mathbf{E}}

\def \G {\mathbf{G}}
\def \h {\bm{h}}

\def \H {\mathbf{H}}
\def \I {\mathbf{I}}
\def \J {\mathbf{J}}

\def \L {\mathcal{L}}

\def \n {\bm{n}}

\def \P {\mathbf{P}}

\def \p {\bm{p}}

\def \R {\mathbf{R}}

\def \T {\mathbf{T}}

\def \w {\bm{w}}
\def \W {\mathbf{W}}

\def \x {\bm{x}}

\def \X {\mathbf{X}}

\def \y {\bm{y}}
\def \Y {\mathbf{Y}}
\def \z {\bm{z}}

\def \Z {\mathbf{Z}}

\begin{document}
%
\title{Simultaneous Sparse Recovery and Blind Demodulation}
%
%
%

\author{Youye~Xie,\IEEEmembership{~Student Member,~IEEE,}
        Michael~B.~Wakin,\IEEEmembership{~Senior Member,~IEEE,}
        and Gongguo~Tang\IEEEmembership{~Member,~IEEE}
        \\
        Department
of Electrical Engineering, Colorado School of Mines, USA
\thanks{The authors are with the Department
of Electrical Engineering, Colorado School of Mines, Golden, CO 80401, USA, e-mails: youyexie@mines.edu, mwakin@mines.edu@mines.edu, gtang@mines.edu.}}

%
%

\markboth{Journal of \LaTeX\ Class Files,~Vol.~14, No.~8, August~2015}%
{Shell \MakeLowercase{\textit{et al.}}: Bare Demo of IEEEtran.cls for IEEE Journals}
%



\maketitle

\begin{abstract}
The task of finding a sparse signal decomposition in an overcomplete dictionary is made more complicated when the signal undergoes an unknown modulation (or convolution in the complementary Fourier domain).
Such simultaneous sparse recovery and blind demodulation problems appear in many applications including medical imaging, super resolution, self-calibration, etc.
In this paper, we consider a more general sparse recovery and blind demodulation problem in which each atom comprising the signal undergoes a {\em distinct} modulation process.
Under the assumption that the modulating waveforms live in a known common subspace, we employ the lifting technique and recast this problem as the recovery of a column-wise sparse matrix from structured linear measurements.
In this framework, we accomplish sparse recovery and blind demodulation simultaneously by minimizing the induced atomic norm, which in this problem corresponds to the block $\ell_{1}$ norm minimization.
For perfect recovery in the noiseless case, we derive near optimal sample complexity bounds for Gaussian and random Fourier overcomplete dictionaries.
We also provide bounds on recovering the column-wise sparse matrix in the noisy case.
Numerical simulations illustrate and support our theoretical results.
\end{abstract}


%
\IEEEpeerreviewmaketitle

\section{Introduction}

\subsection{Overview}

\footnote{Parts of the results in this paper were presented at the 44th International Conference on Acoustics, Speech, and Signal Processing (ICASSP), 2019 \cite{xie2019sparse}.}In classical sparse recovery and compressive sensing problems, a system observes $\y=\D\A\c\in\C^N$ where $\D$, $\A$, and $\c$ are the sensing matrix, dictionary matrix, and sparse signal coefficient vector, respectively. The goal is to recover the sparse vector $\c$ from the measurements $\y$. Usually $\D$ and $\A$ are known, but the whole system is under-determined. This model arises naturally in a wide range of applications such as medical imaging \cite{lustig2008compressed}, seismic imaging \cite{herrmann2008non}, video coding \cite{pudlewski2010performance}, and network traffic monitoring \cite{zhang2009spatio}.

In the special case where $\D$ is diagonal and contains a carrier signal or the Fourier coefficients of a known source signal along in its diagonal entries, $\y$ can be viewed as a modulated version of the signal $\A\c$ \cite{goldsmith2005wireless} or the Fourier transform of the convolution between two source signals \cite{ahmed2014blind}. Recovering $\c$ can thus be viewed as a demodulation (or deconvolution) problem. Unfortunately, in problems like super resolution \cite{yang2016super} and self-calibration \cite{ling2015self}, the modulation matrix $\D$ is unknown a priori, as it incorporates the unknown point spread functions or calibration parameters. Recovering $\D$ and $\c$ jointly is a simultaneous sparse recovery and blind demodulation problem.

In this paper, we consider a more general sparse recovery and blind demodulation problem in which each atom comprising the signal undergoes a {\em distinct} modulation process. Under the assumption that the modulating waveforms live in a known common subspace, we employ the lifting technique and recast this problem as the recovery of a column-wise sparse matrix from structured linear measurements. In this framework, we recover the sparse coefficient vector $\c$ and all of the modulating waveforms simultaneously by minimizing the induced atomic norm~\cite{chandrasekaran2012convex, xie2017radar}, which in this problem corresponds to the block $\ell_{1}$ norm minimization and we also refer to it as the $\ell_{2,1}$ norm minimization.


\subsection{Setup and Notation}

To better illustrate our main contributions and compare to related work, we first define our signal model and the corresponding atomic norm minimization problem.

Throughout this paper, we use bold uppercase, $\X$, bold lowercase, $\x$, and non-bold letters, $x$, to represent matrices, vectors, and scalars. We use $\bar{\cdot}$, $\cdot^H$ and $\cdot^T$ to denote respectively complex conjugate, matrix Hermitian, and matrix transpose. The symbol $C$ denotes a constant. $\X_T$ ($\x_T$, resp.) is a matrix (vector, resp.) that zeros out the columns (entries, resp.) not in $T$. We call $T$ the support of the matrix $\X$ (and vector $\x$), and we use $\tilde{\X}$ to denote the sub-matrix after removing the zero rows or columns in $\X$. $\sign(\x)=\x/||\x||_2$ when $||\x||_2\neq 0$ and $\mathbf{0}$ otherwise. $\sign(\X)=[\sign(\x_1),\cdots, \sign(\x_M)]$. We use $||\cdot||$ to indicate the spectral norm, which returns the maximum singular value of a matrix. The $\ell_{2,1}$ norm of a matrix $\X =[\x_1 \quad\cdots\quad \x_M]$, denoted by $||\X||_{2,1}$, is defined to be $\sum_{j=1}^M||\x_j||_2$. The inner product between vectors and matrices are defined as $\langle \x,\y\rangle = \y^H\x$ and $\langle \X,\Y\rangle = \tr\left(\Y^H\X\right)$ respectively.

\subsection{Problem Formulation}

In this paper, we study a generalized sparse recovery and blind demodulation problem in which the coefficient vector is unknown and each atom (column) of the dictionary undergoes an unknown modulation process. Specifically, we assume the system receives a composite signal
\begin{align}
\y=\sum^{M}_{j=1}c_j\D_j\a_j\in\C^N
\label{origin}
\end{align}
where $c_j\in\C$ is an unknown scalar, $\D_j\in\C^{N \times N}$ is an unknown diagonal modulation matrix, and $\a_j\in\C^{N}$ is the $j$-th atom from a known dictionary $\mathbf{A}= \begin{bmatrix} \a_1 & \a_2 & \cdots & \a_M \end{bmatrix} \in\C^{N\times M}$ with $N<M$. Our goal is to recover both $c_j$ and $\mathbf{D}_j$ for all $j$ from the measurement $\y$.

To make this problem well-posed, among the $M$ over-complete atoms, we assume only $J < M$ of them contribute to the observed signal; that is, at most $J$ coefficients $c_j$ are nonzero. We furthermore assume that each modulation matrix obeys a subspace constraint:
\begin{align}
\mathbf{D}_j = \diag(\mathbf{B}\h_j),
\label{eq:sc}
\end{align}
where $\mathbf{B}\in\C^{N\times K}$ ($N>K$) is a known basis for the $K$-dimensional subspace of possible modulating waveforms, and $\h_j\in\C^K$ is an unknown coefficient vector. Similar subspace assumptions have been made in deconvolution and demixing papers~\cite{ling2017blind, chi2016guaranteed}. With this assumption, recovering $c_j$ and $\mathbf{D}_j$ equals to recovering $c_j$ and $\h_j$. Since $c_j\D_j\a_j = c_j \diag(\mathbf{B}\h_j) \a_j = (k c_j) \diag(\mathbf{B}\left(\frac{1}{k}\h_j\right)) \a_j$ for any $k\neq0$, without loss of generality, we assume $\h_j$ has unit norm and $c_j\geq 0$ with its complex phase and sign absorbed by $\h_j$.

Define $\mathbf{B}^H = [\b'_1\quad \b'_2\quad \cdots\quad \b'_N]\in\C^{K\times N}$ and note that the $n$-th entry of the observed signal can be expressed as
\begin{equation}
\begin{aligned}
\y(n)&=\sum_{j=1}^Mc_j\bar{\a}_j^H\e_n\b_n'^H\h_j=\tr \left(\e_n\b_n'^H\sum_{j=1}^Mc_j\h_j\bar{\a}_j^H\right)\\
    &=\langle\sum_{j=1}^Mc_j\h_j\bar{\a}_j^H,\b'_n\e_n^H\rangle=\langle\G,\b'_n\e_n^H\rangle,
\label{L1}
\end{aligned}
\end{equation}
where $\G=\sum_{j=1}^Mc_j\h_j\bar{\a}_j^H$, and $\e_n$ is the $n$-th column of the $N\times N$ identity matrix. From (\ref{L1}), we see that the measurement vector $\y$ depends linearly on the matrix $\G$ which encodes all of the unknown parameters of interest. We denote this linear sensing process as $\y=\L'(\G)$ and recast the recovery problem as that of recovering $\G$ (and its components) from the linear measurements.

The unknown matrix $\G$ can be viewed as a linear combination of $J$ rank-$1$ matrices from the atomic set $\mathcal{A} := \{\h\bar{\a}^H:\bar{\a} \in\{\bar{\a}_1,...,\bar{\a}_M\},||\h||_2=1\}$ and thus we propose to recover $\G$ using the corresponding atomic norm minimization:
\begin{equation}
\begin{aligned}
&\minimize_{\G\in\C^{K\times N}} ||\G||_{\mathcal{A}}\quad\st \y=\L'(\G).
\label{atomic}
\end{aligned}
\end{equation}
The atomic norm appearing in~\eqref{atomic} is defined as $||\G||_{\mathcal{A}}:=\inf \{\sum_{k}|\tilde{c}_k|:\G=\sum_{k} \tilde{c}_k\J_k, \J_k\in\mathcal{A}\}$. Moreover, the following result establishes its equivalence with the $\ell_{2,1}$ norm.

\begin{Proposition}
The atomic norm optimization problem \eqref{atomic} can be equivalently expressed as the following $\ell_{2,1}$ norm optimization problem

\begin{equation}
\begin{aligned}
&\minimize_{\X\in\C^{K\times M}} ||\X||_{2,1}\quad\st \y=\mathcal{L}(\X)
\label{21}
\end{aligned}
\end{equation}
where $\X=[c_1\h_1\quad c_2\h_2\quad\cdots\quad c_M\h_M]\in\C^{K\times M}$ and $\L$ represents the following linear sensing process
\begin{equation}
\begin{aligned}
\y(n)&=\langle\X,\b'_n\e_n^H\bar{\A}\rangle=\b'^H_n\X \a'_n.
\label{L11}
\end{aligned}
\end{equation}
in which $\b'_n$ and $\a_n'$ are the $n$-th column of $\B^H$ and $\A^T$.
\end{Proposition}

\begin{proof}
We first note that the atomic norm can be equivalently expressed as $||\G||_{\mathcal{A}}=\inf \{\sum_{j=1}^M|c_j|:\G=\sum_{j=1}^M c_j\h_j\bar{\a}_j^H,||\h_j||_2=1\}$. To see this, consider any decomposition of $\G$ of the form $\G=\sum_k \tilde{c}_k\J_k$ with $\J_k \in \mathcal{A}$. Define $\mathcal{N}_j=\{k:\J_k=\tilde{\h}_k\bar{\a}_j^H\}$ and write $\G = \sum_{j=1}^M(\sum_{k\in\mathcal{N}_j}\tilde{c}_k\tilde{\h}_k)\bar{\a}_j^H$. This is equivalent to writing $\G=\sum_{j=1}^M c_j\h_j\bar{\a}_j^H$ where $\h_j=\frac{\sum_{k\in\mathcal{N}_j}\tilde{c}_k\tilde{\h}_k}{||\sum_{k\in\mathcal{N}_j}\tilde{c}_k\tilde{\h}_k||_2}$ and $c_j=||\sum_{k\in\mathcal{N}_j}\tilde{c}_k\tilde{\h}_k||_2$. Finally, note that $|c_j|\leq \sum_{k\in\mathcal{N}_j}|\tilde{c}_k|$.

Next, to establish the equivalence with the $\ell_{2,1}$ norm, for any $c_j$ and $\h_j$ with $||\h_j||_2=1$, define $\mathbf{x}_j=c_j\mathbf{h}_j$ and $\X =[\mathbf{x}_1\quad \mathbf{x}_2\quad\cdots\quad \mathbf{x}_M]$. Then
\begin{equation*}
\begin{aligned}
&\quad||\G||_{\mathcal{A}}\\
&=\inf \Bigg\{\sum_{j=1}^M|c_j|:\G=\sum_{j=1}^Mc_j\h_j\bar{\a}_j^H,||\h_j||_2=1\Bigg\}\\
                   &=\inf \Bigg\{\sum_{j=1}^M||\x_j||_2:\G=\sum_{j=1}^M\x_j\bar{\a}_j^H\Bigg\}\\
                   &=\inf \big\{||\X||_{2,1}:\G=\X\bar{\A}^H\big\}.
\end{aligned}
\end{equation*}
Finally, to establish the equivalence of the linear sensing process, (\ref{L1}) indicates that for $\G = \X\bar{\A}^H$,
\begin{equation*}
\begin{aligned}
\y(n)&=\langle\G,\b'_n\e_n^H\rangle=\langle\X,\b'_n\e_n^H\bar{\A}\rangle=\b'^H_n\X \a'_n.
\end{aligned}
\end{equation*}
\end{proof}

The above optimization focuses on recovering the structured matrix $\X$ from linear measurements. Once the optimization is solved, the unknown parameters can be easily extracted from the solution $\hat{\X}$ as follows:
\begin{align}
c_j=||\hat{\x}_j||_2, \quad\h_j=\frac{\hat{\x}_j}{||\hat{\x}_j||_2}, \text{ and } \D_j=\text{diag}(\B\h_j)
\label{recover}
\end{align}
for $\hat{\x}_j\neq0$ and $1\leq j\leq M$.

The adjoint of the linear operator $\L$ is $\L^*(\y)=\sum_{l=1}^Ny_l\b'_l\a'^H_{l}$. The linear operator $\L$ also has a matrix-vector multiplication form. Note that $\L(\X)=\mathbf{\Phi}\cdot\text{vec}(\X)$, where $\mathbf{\Phi}\in\C^{N\times KM}$ is
\begin{equation}
\begin{aligned}
\mathbf{\Phi}=[ \mathbf{\phi}_{1,1}\quad\cdots\quad\pphi_{K,1}\quad\cdots\quad\pphi_{1,M}\quad\cdots\quad\pphi_{K,M} ]
\label{phi}
\end{aligned}
\end{equation}
in which $\pphi_{i,j}=\text{diag}(\b_i)\a_j\in\C^{N\times 1}$ and $\b_i$ is the $i$-th column of $\B$. Furthermore,
\begin{align}
\PPhi^H = [\pphi'_1\quad\pphi'_2\quad \cdots\quad\pphi'_N]\in\C^{KM\times N}
\label{phiH}
\end{align}
where $\mathbf{\phi}'_i=\bar{\a}'_i\otimes \b_i'\in\C^{KM\times1}$.

Finally, we note that the observed signal could be contaminated with noise. In this case, our measurement model becomes
\begin{align}
\y=\sum^{M}_{j=1}c_j\D_j\a_j+\n
\label{originnoisy}
\end{align}
for some unknown noise vector $\n\in\C^{N\times1}$ which we suppose satisfies $||\n||_2\leq \eta$. In this case, we can write $\y=\L(\X_0)+\n$, where $\X_0$ is the ground truth solution. As an alternative to equality-constrained $\ell_{2,1}$ norm minimization~\eqref{21}, we then consider the following relaxation:
\begin{equation}
\begin{aligned}
&\minimize_{\X\in\C^{K\times M}} ||\X||_{2,1}\quad\st ||\y-\L(\X)||_2\leq \eta.
\label{noisy}
\end{aligned}
\end{equation}

\subsection{Applications of The Proposed Signal Model\label{app}}

The proposed signal model encompasses a wide range of applications. We briefly introduce some of them as follows.

\subsubsection{Direction of arrival estimation for antenna array\label{DOA}}

We first consider the direction of arrival (DOA) estimation problem in antenna array. Assume we have a linear array antenna consisting of $N$ elements, and we want to estimate the DOAs of several sources from a snapshot of the received signal. In addition, we consider the narrowband scenario and confine the the array and the far-field sources to a common plane as described in \cite{malioutov2005sparse}. In this case, the DOA is determined by the azimuth angle, $\theta$, of the source, which ranges from $0$ to $180$ degrees. Mathematically, after discretizing the azimuth angle into $M$ grids, the observervation of the array can be represented as \cite{friedlander2014bilinear}
\begin{align*}
\y=\D\A(\theta)\c+\n\in\C^{N\times1}
\end{align*}
where $\D\in\C^{N\times N}$ is the diagonal matrix capturing the unknown calibration of the array elements \cite{ling2015self}. Particularly, the calibration issue may arise from gain discrepancies caused by the change of temperatures and humidity of the environment \cite{ling2015self}. Namely, the channel is not ideal. One can simulate different scenarios and collect many possible calibration vectors. By applying the singular value decomposition (SVD) on the matrix formed by those calibration vectors, we can then extract the subspace matrix, $\B$, with desired dimensions to approximate the calibration using $\D=\diag(\B\h)$ where $\h$ is the unknown coefficient vector. $\A(\theta)\in\C^{N\times M}$ is the known array manifold matrix whose columns $\a(\theta_j)$ for $j\in\{1,2,\cdots,M\}$ are the steering vectors. For uniformly spaced linear array antenna (ULA), $\a(\theta_j)=[1,e^{i\frac{2\pi d}{\lambda}\cos(\theta_j)},\cdots,e^{i(N-1)\frac{2\pi d}{\lambda}\cos(\theta_j)}]$ where $d$ is the distance between array elements and $\lambda$ is the radar operating wavelength \cite{ward1998space}. Moreover, the entries of $\c$ indicate the strength of the impinging signals and if there exists $J(<M)$ sources, only $J$ entries of $\c$ are nonzero. $\n$ consists of the discretization error, approximation error, and additive noise.

Furthermore, let us consider a more severe while realistic situation, where the calibration is sensitive to the direction of arrival which implies that the channel responses from different angles are slightly different. So that the calibration matrix, $\D$, are different for different $\theta_j$. In this case, we can write
\begin{align*}
\y=\sum^{M}_{j=1}c_j\D_j\a(\theta_j)+\n\in\C^{N\times1}.
\end{align*}

\subsubsection{Super-resolution for single molecule imaging\label{smi}}

Another application is the single molecule imaging \cite{tang2013sparse} via stochastic optical reconstruction microscopy (STORM) \cite{rust2006sub}. In this application, the cellular structure of the object of interest is dyed with fluorophores, and STORM divides the imaging process into thousands of cycles. Within each cycle or observation, only a portion of the fluorophores are activated and imaged. Therefore, a typical observation is a low-resolution frame with its activated fluorophores convolved with the non-stationary point spread functions of the microscope, which can be represented as
\begin{align*}
\y=Sample\left[\sum_{j=1}^M c_j(\B'\h_j)\circledast\e_j +\n'\right] \in\R^{N\times 1}
\end{align*}
where $\y\in\R^{N\times 1}$ is a vectorized, imaged frame downsampled from its super-resolution image with $M (>N)$ pixels, $c_j$ represents the intensity of the activated fluorophores, and $\B'$ is the subspace that the point spread functions live in. $\e_j\in\R^M$, which indicates the location of the activated flurophores, is the $j$-th column of the identity matrix and $\n'$ denotes the noise. Moreover, $\y$ can also be represented equivalently as
\begin{align*}
\y=Sample\Bigg\{IDFT\left[\sum_{j=1}^M c_j\D_j\a_j +\n\right]\Bigg\} \in\R^{N\times 1},
\end{align*}
where $IDFT[\cdot]$ is the inverse discrete Fourier transform (DFT) operator, $\D_j=\diag(\B\h_j)$ with $\B=DFT[\B']$, and $\a_j$s are the DFT of spikes containing the location information. $\n=DFT[\n']$. The goal of this application is to recover the super-resolution image from its low-resolution frame $\y$, or mathematically, locating the nonzero $c_j$.

\vspace{2mm}

Other applications that fit into the model investigated in this work include frequency estimation with damping that appears in nuclear magnetic resonance spectroscopy \cite{cai2016robust} with damping signals approximately living in a common subspace \cite{yang2016super} and the CDMA system with spreading sequence sensitive channel as described in Section 6.4 of \cite{ling2015self}.

\subsection{Main Contributions}

Our contributions are twofold.
First, we employ $\ell_{2,1}$ norm minimization to achieve sparse recovery and blind demodulation simultaneously given the generalized signal model from equation (\ref{origin}).
Second, for perfect recovery of all parameters in the noiseless case, we derive near optimal sample complexity bounds for the cases where $\A$ is a random Gaussian and a random subsampled Fourier dictionary.
Both of bounds require the number of measurements $N$ to be proportional to the number of degrees of freedom, $O(JK)$, up to log factors.
We also provide bounds on recovering the column-wise sparse matrix in the noisy case; these bounds show that the recovery error scales linearly with respect to the strength of the noise.

\subsection{Related Work}

The $\ell_{2,1}$ norm has been widely used to promote sparse recovery in multiple measurement vector (MMV) problems \cite{cotter2005sparse,chen2005sparse}. The MMV problem involves a collection of sparse signal vectors that are stacked as the rows of a matrix $\X$. These signals have a common sparsity pattern, which results in a column-wise sparse structure for $\X$. As
in our setup, the $\ell_{2,1}$ norm is used to recover $\X$ from linear measurements of the form $\y =\PPhi_{MMV} \cdot\text{vec}(\X^T)$. However, $\PPhi_{MMV}$ has a block diagonal structure where all diagonal sub-matrices are the same which is the dictionary matrix. This is different from the structure of the linear measurements in our problem; see for example~\eqref{phi}.

Our work is also closely related to certain recent works in model-based deconvolution, self-calibration, and demixing.
When all $\D_j$ in (\ref{origin}) are the same, our signal model coincides with the self-calibration problem in \cite{ling2015self}, although that work employs $\ell_1$ norm minimization rather than $\ell_{2,1}$ norm minimization to recover $\X$.
A more recent paper \cite{hung2017low} does apply the $\ell_{2,1}$ norm for the self-calibration problem but again assumes a common modulation matrix $\D$.
The paper \cite{ling2017blind} generalizes the work of \cite{ling2015self} and considers a blind deconvolution and demixing problem which can be interpreted as the self-calibration scenario with multiple sensors whose calibration parameters might be different. However, the signal model in that paper is not directly comparable to our model, and the recovery approach studied in that paper involves nuclear norm minimization and requires knowledge of the number of sensors.
A blind sparse spike deconvolution is studied in \cite{chi2016guaranteed}, wherein the dictionary consists of sampled complex sinusoids over a continuous frequency range and all atoms undergo the same modulation. Inspired by \cite{chi2016guaranteed}, \cite{yang2016super} generalizes the model to the case of different modulating waveforms. Like \cite{chi2016guaranteed}, however, \cite{yang2016super} also considers a sampled sinusoid dictionary over a continuous frequency range, and it employs a random sign assumption on the coefficient vectors $\h_j$ which makes it difficult to derive recovery guarantees with noisy measurements. More works considering a common modulation process can be found in \cite{ahmed2014blind, eldar2017sensor, flinth2018sparse}.

Our work can be viewed as a generalization of the self-calibration \cite{ling2015self} and blind deconvolution problems \cite{ahmed2014blind}. Moreover, our analysis is quite different from the works considering the continuous sinusoid dictionary \cite{chi2016guaranteed, yang2016super}, since the tools in those papers are specialized to the continuous sinusoids dictionary and we consider discrete Gaussian and random Fourier dictionaries in both noiseless and noisy settings.

\vspace{2mm}

The rest of the paper is organized as follows. In Section~\ref{sec:mainresults}, we present our main theorems regarding perfect parameter recovery in the noiseless setting and matrix denoising in the noisy setting. Sections~\ref{sec:proof1} and~\ref{sec:proof2} contain the detailed proofs of the main theorems. Several numerical simulations are provided in Section~\ref{sec:sims} to illustrate the critical scaling relationships, and we conclude in Section~\ref{sec:conclusion}.

\section{Main Results}
\label{sec:mainresults}

We present our main theorems in this section. In each of the noisless and noisy cases, we consider two models for the dictionary matrix $\A$. In the first model, $\A\in\R^{N\times M}$ is a real-valued random Gaussian matrix, with each entry sampled independently from the standard normal distribution. In the second model, $\A\in \C^{N\times M}$ is a complex-valued random Fourier matrix, with each of its $N (< M)$ rows chosen uniformly with replacement from the $M\times M$ discrete Fourier transform matrix $\F$ where $\F^H\F=M\I_M$. Our first theorem concerns perfect parameter recovery in the noiseless setting.

\begin{Theorem}
(Noiseless case) Consider the measurement model in equation (\ref{origin}), assume that at most $J (< M)$ coefficients $c_j$ are nonzero, and furthermore assume that the nonzero coefficients $c_j$ are real-valued and positive. Suppose that each modulation matrix $\D_j$ satisfies the subspace constraint~\eqref{eq:sc}, where $\B^H\B=\I_K$ and each $\h_j$ has unit norm \footnote{Theorem \ref{maintheom1} actually works for $\h_j$ with arbitrary norms as long as the relative scale between $c_j$ and $\h_j$ is known. }.

Then the solution $\hat{\X}$ to problem \eqref{21} is the ground truth solution $\X_0$---which means that $c_j$, $\h_j$, and $\D_j$ can all be successfully recovered for each $j$ using \eqref{recover}---with probability at least $1-O(N^{-\alpha+1})$
\begin{itemize}
\item if $\A\in\R^{N\times M}$ is a random Gaussian matrix and
\begin{align}
\frac{N}{\log^2(N)}\geq C_{\alpha}\mu_{max}^2KJ(\log(M-J)+\log(N)).
\label{maintheom11}
\end{align}
\item if $\A\in\C^{N\times M}$ is a random Fourier matrix and
\begin{equation}
\begin{aligned}
N\geq &\; C_{\alpha}\mu_{max}^2KJ\log(4\sqrt{2J}\gamma)\cdot\\
&\qquad\qquad(\log(M-J)+\log(K+1)+\log(N))
\label{maintheom12}
\end{aligned}
\end{equation}
where $\gamma=\sqrt{2M\log(2KM)+2M+1}$.
\end{itemize}
In both cases, $C_\alpha$ is a constant defined for $\alpha> 1$ and the coherence parameter
\begin{align*}
\mu_{max} = \max_{i,j}\sqrt{N}|\B_{ij}|.
\end{align*}
\label{maintheom1}
\end{Theorem}

We note that both of the sample complexity bounds in Theorem~\ref{maintheom1} require the number of measurements $N$ to be proportional to the number of degrees of freedom, $O(KJ)$, up to log factors. We also note that the sample complexity bounds scale with the square of the coherence parameter $\mu_{max} = \max_{i,j}\sqrt{N}|\B_{ij}|$. Under the assumption $\B^H\B=\I_K$ which requires the columns of $\B$ to be orthonormal, $\mu_{max}\in[1,\sqrt{N}]$. Specifically, given the system parameters with large enough $N$, \eqref{maintheom11} is satisfied when $1\leq\mu_{max}\leq\sqrt{\frac{N}{C_{\alpha}\log^2(N)KJ(\log(M-J)+\log(N))}}$. The valid range of $\mu_{max}$ for \eqref{maintheom12} and the noisy case can be easily derived in the same manner. And $\mu_{max}$ is minimized when the energy of each column of $\B$ is not concentrated on a few entries but spread across the whole column.

Our second theorem provides bounds on recovering the column-wise sparse matrix in the noisy case; these bounds show that the recovery error scales linearly with respect to the strength of the noise.

\begin{Theorem}
(Noisy case)
Consider the measurement model in equation (\ref{originnoisy}), assume that at most $J (< M)$ coefficients $c_j$ are nonzero, and furthermore assume that the norm of the noise is bounded, $||\n||_2\leq\eta$. Suppose also that each modulation matrix $\D_j$ satisfies the subspace constraint~\eqref{eq:sc}, where $\B^H\B=\I_K$.

Then with probability at least $1-O(N^{-\alpha+1})$, the solution $\hat{\X}$ to problem \eqref{noisy} satisfies
\begin{itemize}
\item if $\A\in\R^{N\times M}$ is a random Gaussian matrix,
\begin{align}
||\hat{\X}-\X_0||_F\leq\left(C_1+C_2\sqrt{J}\right)\eta
\label{gnoisy}
\end{align}
when
\begin{equation}
\begin{aligned}
\frac{N}{\log^2(N)}\geq &\ C_{\alpha}\mu_{max}^2KJ\left(\log(C\mu_{max}\sqrt{KJ})C+1\right)\cdot\\
&\quad\qquad(\log(M-J)+\log(MK)+\log(N))
\label{maintheom21}
\end{aligned}
\end{equation}
where $C$ is a constant.
\item if $\A\in\C^{N\times M}$ is a random Fourier matrix,
\begin{align}
||\hat{\X}-\X_0||_F\leq\left(C_1+C_2\sqrt{PJ}\right)\eta
\label{fnoisy}
\end{align}
when
\begin{equation}
\begin{aligned}
N\geq &\; C_{\alpha}\mu_{max}^2KJ\log(4\sqrt{2J}\gamma)\cdot\\
&\qquad\qquad(\log(M-J)+\log(MK)+\log(N))
\label{maintheom22}
\end{aligned}
\end{equation}
where $\gamma=\sqrt{2M\log(2KM)+2M+1}$ and $P\geq \log(4\sqrt{2J}\gamma)/\log2$.
\end{itemize}
 In both cases, $C_\alpha$ is defined for $\alpha> 1$. $C_1$ and $C_2$ are constant.
\label{maintheom2}
\end{Theorem}

Although Theorem~\ref{maintheom2} focuses exclusively on bounding the recovery error of the matrix $\X_0$, one can also attempt to estimate the parameters $c_j$, $\h_j$, and $\D_j$ from $\hat{\X}$ using \eqref{recover}. And according to Theorem \ref{maintheom2}, for any $\hat{\x}_j=\hat{c}_j\hat{\h}_j$ and $\x_{0,j}=c_{0,j}\h_{0,j}$ where $\hat{\x}_j$ and $\x_{0,j}$ are the $j$-th columns of the solution $\hat{\X}$ and the ground truth $\X_0$ respectively, we would have $||\hat{c}_j\hat{\D}_j-c_{0,j}\D_{0,j}||_F=||\hat{c}_j\hat{\h}_j-c_{0,j}\h_{0,j}||_2\leq \left(C_1+C_2\sqrt{J}\right)\eta$ with random Gaussian dictionary and $||\hat{c}_j\hat{\D}_j-c_{0,j}\D_{0,j}||_F=||\hat{c}_j\hat{\h}_j-c_{0,j}\h_{0,j}||_2\leq  \left(C_1+C_2\sqrt{PJ}\right)\eta$ for random Fourier dictionary. In addition, as results on structured matrix recovery from (possibly noisy) linear measurements, we believe that Theorems~\ref{maintheom1} and~\ref{maintheom2} may be of independent interest outside of the sparse recovery and blind demodulation problem.

\section{Proof of Theorem \ref{maintheom1}}
\label{sec:proof1}

To begin our proof of the main theorem in the noiseless case, we first derive sufficient conditions for exact recovery.

\subsection{Sufficient Conditions for Exact Recovery}

Sufficient conditions for exact recovery are the null space property and an alternative sufficient condition derived from the null space property. Similar sufficient conditions with complete proofs are available  for  minimization problems using other types of norms \cite{ling2015self, candes2011probabilistic, foucart2013mathematical, vaiter2015model}. However, since we cannot find sufficient conditions that suit our purpose and in order to be self-contained, we provide a short proof for the ones specific to the $\ell_{2,1}$ norm minimization problem in this section.

\begin{Proposition}
(The null space property) The matrix $\X_0=[c_1\h_1\quad c_2\h_2\quad....\quad c_M\h_M]\in\C^{K\times M}$ with support $T$ is the unique solution to the inverse problem (\ref{21}) if
\begin{align*}
-|\langle\H_T,\sign(\X_0)\rangle|+||\H_{T^C}||_{2,1}>0
\end{align*}
for any $\H\neq\textbf{0}$ in the nullspace of $\L$.
\label{pro1}
\end{Proposition}
\begin{proof}
Let $\hat{\X}=\X_0+\H$ be a solution to problem (\ref{21}), with $\L(\H)=\mathbf{0}$. To prove $\X_0$ is the unique solution, it is sufficient to show that $||\hat{\X}||_{2,1}>||\X_0||_{2,1}$ if $\H\neq\textbf{0}$. We start by observing that
\begin{equation*}
\begin{aligned}
&\quad||\X_0+\H||_{2,1}=||\X_{0,T}+\H_T||_{2,1}+||\H_{T^C}||_{2,1}\\
               &\geq |\langle\X_{0,T}+\H_T,\sign(\X_{0,T})\rangle|+||\H_{T^C}||_{2,1}\\
               &=|\langle\X_{0,T},\sign(\X_{0,T})\rangle+\langle\H_T,\sign(\X_{0,T})\rangle|+||\H_{T^C}||_{2,1}\\
               &\geq||\X_{0,T}||_{2,1}-|\langle\H_T,\sign(\X_{0,T})\rangle|+||\H_{T^C}||_{2,1}
\end{aligned}
\end{equation*}
where $\sign(\X_{0,T})=\sign(\X_0)$ and the first inequality comes from the fact that
\begin{equation}
\begin{aligned}
&\quad||\X_{0,T}+\H_T||_{2,1}=\sum_{i\in T}||\x_{0,i}+\h_i||_2||\sign(\x_{0,i})||_2\\
                   &\geq\sum_{i\in T}|\langle\x_{0,i}+\h_i,\sign(\x_{0,i})\rangle|\geq|\langle\X_{0,T}+\H_T,\sign(\X_{0,T} \rangle|.
\label{ineq}
\end{aligned}
\end{equation}
Therefore, as long as $-|\langle\H_T,\sign(\X_0)\rangle|+||\H_{T^C}||_{2,1}>0$ for any $\H\neq \textbf{0}$ in the nullspace of $\L$, $\X_0$ is the unique solution.
\end{proof}

\begin{Proposition}
The matrix $\X_0\in\C^{K\times M}$ with support $T$ is the unique solution to the inverse problem \eqref{21} if there exists $\gamma > 0$ and a matrix $\Y$ in the range space of $\L^*$ such that
\begin{align*}
||\Y_T-\sign(\X_{0,T})||_{F}\leq \frac{1}{4\sqrt2\gamma}\quad \text{and} \quad ||\Y_{T^C}||_{2,\infty}\leq\frac{1}{2}
\end{align*}
and the operator $\L$ satisfies ($\L_T(\X) = \{ \b'^H_n\X \a'_{n,T}\}_{n=1}^N$)
\begin{align}
||\L_{T}^*\L_{T}-\I_T||\leq\frac{1}{2}\quad \text{and} \quad ||\L||\leq \gamma.
\label{eq:isomopnorm}
\end{align}
\label{inexact}
\end{Proposition}

\begin{proof}
Proposition \ref{pro1} shows that to establish uniqueness, it is sufficient to prove that $-|\langle\H_T,\sign(\X_0)\rangle|+||\H_{T^C}||_{2,1}>0$ for any $\H\neq \textbf{0}$ in the nullspace of $\L$. Note that
\begin{equation*}
\begin{aligned}
&-|\langle\H_T,\sign(\X_0)\rangle|+||\H_{T^C}||_{2,1}\\
               &~=-|\langle\H_T,\sign(\X_0)-\Y_T\rangle+\langle\H_T,\Y_T\rangle|   +||\H_{T^C}||_{2,1}\\
               &~\geq-|\langle\H_T,\sign(\X_0)-\Y_T\rangle|-|\langle\H_{T^C},\Y_{T^C}\rangle|
               +||\H_{T^C}||_{2,1}
\end{aligned}
\end{equation*}
since $\langle\H_T,\Y_T\rangle=-\langle\H_{T^C},\Y_{T^C}\rangle$. By applying the H\"{o}lder inequality, we get a stronger condition
\begin{align*}
-||\sign(\X_0)-\Y_T||_F||\H_T||_F+(1-||\Y_{T^C}||_{2,\infty})||&\H_{T^C}||_{2,1} \\
&>0.
\end{align*}
Since $||\L_{T}^*\L_{T}-\I_T||\leq\frac{1}{2}$ and $||\L||\leq \gamma$, we have $||\L(\H_T)||_F\geq\frac{1}{\sqrt2}||\H_T||_F$, $||\L(\H_{T^C})||_F\leq\gamma||\H_{T^C}||_F$ and
\begin{equation}
\begin{aligned}
\frac{1}{\sqrt2}||\H_T||_F\leq||\L(\H_T)||_F&=||\L(\H_{T^C})||_F\\
&\leq\gamma||\H_{T^C}||_F\leq\gamma||\H_{T^C}||_{2,1}.
\label{ine}
\end{aligned}
\end{equation}
Plugging~\eqref{ine} into the stronger condition above yields
\begin{align*}
\left(1-||\Y_{T^C}||_{2,\infty}-\sqrt2\gamma||\sign(\X_0)-\Y_T||_F\right)||\H_{T^C}||_{2,1}>0.
\end{align*}
Therefore, if $||\Y_T-\sign(\X_{0,T})||_{F}\leq \frac{1}{4\sqrt2\gamma}$, $||\Y_{T^C}||_{2,\infty}\leq\frac{1}{2}$, and $\H_{T^C}\neq\textbf{0}$, the left hand side is positive. On the other hand, if $\H_{T^C}=\textbf{0}$, from (\ref{ine}), $\H_T=\textbf{0}$ and $\H=\textbf{0}$.
\end{proof}

\subsection{Bounding The Isometry Constant and Operator Norm}

In this section, we bound the isometry constant and operator norm $\gamma$ appearing in~\eqref{eq:isomopnorm} based on the randomness in the matrix $\A$. The isometry bound for the linear operator $\L$ can be found in Lemma 4.3 in \cite{ling2015self}.

\begin{Lemma}\cite[Lemma 4.3]{ling2015self}
(Isometry) For the linear operator $\L$ defined in (\ref{21}) with $\B^H\B=\I_K$ and $\delta>0$,
\begin{align*}
||\PPhi^H_T\PPhi_T-\I_{T}||=||\L_{T}^*\L_{T}-\I_T||\leq\delta
\end{align*}
with probability at least $1-N^{-\alpha+1}$ where $\I_T$ is the identity operator on the support $T$ such that $\I_T(\X)=\X_T$,
\begin{itemize}
\item if $\A$ is a random Gaussian matrix and $N\geq C_{\alpha}\mu_{max}^2KJ\max\{\log(N)/\delta^2,\log^2(N)/\delta\}$.
\item if $\A$ is a random Fourier matrix and $N\geq C_{\alpha}\mu_{max}^2KJ\log(N)/\delta^2$.
\end{itemize}
Here $C_{\alpha}$ is a constant that grows linearly with $\alpha>1$.
\label{isometryproperty}
\end{Lemma}

$\L_T(\X)=\PPhi_T\cdot\text{vec}(\X)$ and $\PPhi_T$ can be viewed as $\PPhi$ constructed using $\A_T$, whose $i$-th column is zero if $i\in\T^C$, following \eqref{phi}. Therefore, $\PPhi_T\in\C^{N\times KM}$ has many zero columns and removing those zero columns results in $\tilde{\PPhi}_T\in\C^{N\times KJ}$. If $||\PPhi^H_T\PPhi_T-\I_{T}||=||\tilde{\PPhi}^H_T\tilde{\PPhi}_T-\tilde{\I}_{T}||\leq\delta<1$, $\tilde{\PPhi}^H_T\tilde{\PPhi}_T$ is invertible and
$||(\tilde{\PPhi}^H_T\tilde{\PPhi}_T)^{-1}||\leq (1-\delta)^{-1}$ according to Lemma A.12 in \cite{foucart2013mathematical}. This property will be applied in \eqref{gaussianconstruction} and Theorem \ref{gaussian}. To bound the operator norm of $\L$, we use results from~\cite{ahmed2014blind} and \cite{ling2015self}.

\begin{Lemma}\cite{ahmed2014blind,ling2015self}
For the linear operator $\L$ defined in \eqref{21} with $\B^H\B=\I_K$ and $\alpha\geq1$,
\begin{itemize}
\item if $\A$ is a random Gaussian matrix,
\begin{align*}
||\L||\leq\sqrt{M\log(MN/2)+\alpha \log(N)}
\end{align*}
with probability at least $1-N^{-\alpha}$.

\item if $\A$ is a random Fourier matrix,
\begin{align*}
||\L||\leq\sqrt{2M\log(2KM)+2M+1}
\end{align*}
with probability at least $1-N^{-\alpha}$ when $N\geq \alpha\mu_{max}^2K\log(N)$.

\end{itemize}
\label{operatornorm}
\end{Lemma}

\subsection{Constructing The Dual Certificate for The Gaussian Case}

In the case where $\A$ is a random Gaussian matrix, we construct a certificate matrix $\Y$ that satisfies the conditions in Proposition \ref{inexact}. When $||\PPhi^H_T\PPhi_T-\I_{T}||\leq\frac{1}{2}$, we can set
\begin{align}
\text{vec}(\Y) = \PPhi^H\p = \text{vec}(\L^*(\p))\in\C^{KM\times1},
\label{eq:ydual}
\end{align}
where
\begin{align}
\p = \tilde{\PPhi}_T(\tilde{\PPhi}_T^H\tilde{\PPhi}_T)^{-1}\text{vec}(\sign(\tilde{\X}_{0,T}))\in\C^{N\times 1}.
\label{gaussianconstruction}
\end{align}
By construction, $\Y_T=\sign(\X_{0,T})$, and we need only to verify that $||\Y_{T^C}||_{2,\infty}\leq 1/2$.

\begin{Theorem}
If $||\PPhi^H_T\PPhi_T-\I_{T}||\leq\frac{1}{2}$, there exists $\Y$ in the range space of $\L^*$ such that
\begin{align*}
\Y_T=\sign(\X_{0,T})\quad \text{and} \quad ||\Y_{T^C}||_{2,\infty}\leq\frac{1}{2}
\end{align*}
with probability at least $1-(M-J)e^{-\alpha}$ when $N\geq 40\alpha\mu_{max}^2KJ$ for $\alpha\geq\log(M-J)$.
\label{gaussian}
\end{Theorem}

\begin{proof}
To simplify the notation, without loss of generality, we assume the support of $\X_0$ is the first $J$ columns. Let $\Y$  be the dual certificate matrix defined in~\eqref{eq:ydual}. After removing the columns of $\Y$ on support $T$, we obtain $\text{vec}(\tilde{\Y}_{T^C})\in \C^{K(M-J)\times 1}$ which takes the form
\begin{equation*}
\begin{aligned}
&\text{vec}(\tilde{\Y}_{T^C})=\tilde{\PPhi}^H_{T^C}\p\\
&=[\phi_{1,J+1}^H\p,\cdots,\phi_{K,J+1}^H\p,\phi_{1,J+2}^H\p,\cdots,\phi_{K,M}^H\p]^T\\
&=[\a_{J+1}^H\text{diag}(\bar{\b}_1)\p,\cdots,\a_{J+1}^H\text{diag}(\bar{\b}_K)\p,\\
&\qquad \qquad \qquad \a_{J+2}^H\text{diag}(\bar{\b}_1) \p,\cdots,\a_{M}^H\text{diag}(\bar{\b}_K)\p]^T.
\end{aligned}
\end{equation*}
The columns of $\tilde{\PPhi}_{T^C}$ are independent of $\p$ since $\p$ is constructed with $\a_i$ $(i\in T)$. Equivalently,
\begin{align*}
\tilde{\Y}_{T^C} =
\begin{bmatrix}
\a_{J+1}^H\text{diag}(\bar{\b}_1)\p & \cdots & \a_{M}^H\text{diag}(\bar{\b}_1)\p\\
\a_{J+1}^H\text{diag}(\bar{\b}_2)\p & \cdots & \a_{M}^H\text{diag}(\bar{\b}_2)\p\\
\vdots &  \ddots & \vdots\\
\a_{J+1}^H\text{diag}(\bar{\b}_K)\p & \cdots & \a_{M}^H\text{diag}(\bar{\b}_K)\p
\end{bmatrix}.
\end{align*}
Thus $||\Y_{T^C,j}||_2=||\P\a_{j}||_2$ $(j>J)$ where $\a_{j}$ is real and
\begin{align*}
\P =
\begin{bmatrix}
\p^T\text{diag}(\bar{\b}_1)\\
\p^T\text{diag}(\bar{\b}_2)\\
\vdots \\
\p^T\text{diag}(\bar{\b}_K)
\end{bmatrix}\in\C^{K\times N}.
\end{align*}
We set $\mathbf{\Sigma}=\P^H\P\in\C^{N\times N}$ and have
\begin{align*}
\tr(\mathbf{\Sigma})=||\P||_F^2\leq\frac{2\mu_{max}^2KJ}{N}
\end{align*}
since each row of $\P$ can be bounded by
\begin{equation*}
\begin{aligned}
&\quad||\p^T\text{diag}(\bar{\b}_k)||_2^2\leq\frac{\mu_{max}^2}{N}||\p||_2^2\\
                                   &=\frac{\mu_{max}^2}{N}\text{vec}(\sign(\tilde{\X}_{0,T}))^H(\tilde{\PPhi}_T^H\tilde{\PPhi}_T)^{-1}\text{vec}(\sign(\tilde{\X}_{0,T}))\\
                                   &\leq \frac{2\mu_{max}^2}{N}||\sign(\tilde{\X}_{0,T})||_F^2= \frac{2\mu_{max}^2J}{N}
\end{aligned}
\end{equation*}
since we assume $||\PPhi^H_T\PPhi_T-\I_{T}||\leq\frac{1}{2}$ which implies $||(\tilde{\PPhi}_T^H\tilde{\PPhi}_T)^{-1}||\leq 2$. By generalizing Proposition 1 in \cite{hsu2012tail} to our case, we have
\begin{align*}
\text{Pr}\left( ||\P\a_j||_2^2>\tr(\mathbf{\Sigma})+2\sqrt{\tr(\mathbf{\Sigma}^2)\alpha}+2||\mathbf{\Sigma}||\alpha \right)\leq e^{-\alpha}.
\end{align*}
In addition, because $\mathbf{\Sigma}$ is positive semi-definite and all its eigenvalues are non-negative, $\tr(\mathbf{\Sigma}^2)=\sum_{i=1}^N\lambda_i^2\leq(\sum_{i=1}^N\lambda_i)^2= \tr(\mathbf{\Sigma})^2$ where $\lambda_i$ is the $i$-th eigenvalue of $\mathbf{\Sigma}$. $||\mathbf{\Sigma}||=\sigma_{max}\leq\sum_{i=1}^N\lambda_i=\tr(\mathbf{\Sigma})$ where $\sigma_{max}$ is the maximum singular value of $\mathbf{\Sigma}$. Therefore, for $\alpha>1$, we obtain
\begin{equation*}
\begin{aligned}
&\quad\tr(\mathbf{\Sigma})+2\sqrt{\tr(\mathbf{\Sigma}^2)\alpha}+2||\mathbf{\Sigma}||\alpha \\
&\leq \tr(\mathbf{\Sigma})+2\tr(\mathbf{\Sigma})\alpha+2\tr(\mathbf{\Sigma})\alpha\leq \frac{2\mu_{max}^2KJ}{N}(1+4\alpha).
\end{aligned}
\end{equation*}
If we pick $N\geq40\alpha\mu_{max}^2KJ$, $||\P\a_j||_2>1/2$ with probability at most $e^{-\alpha}$. Taking the union over all $(M-J)$ non-zero columns of $\Y_{T^C}$ gives
\begin{align*}
\text{Pr}(||\Y_{T^C}||_{2,\infty}>1/2)&\leq (M-J)e^{-\alpha}.
\end{align*}
Therefore, $||\Y_{T^C}||_{2,\infty}\leq 1/2$ with probability at least $1-(M-J)e^{-\alpha}$ when $N\geq 40\alpha\mu_{max}^2KJ$. To make the probability meaningful, $\alpha$ should be greater than $\log(M-J)$.

\end{proof}

\subsection{Proof of Theorem \ref{maintheom1} for Random Gaussian Dictionary\label{gsum}}

In this section, we assemble the pieces to complete the proof of Theorem \ref{maintheom1} in the Gaussian case. To do so, we ensure that all sufficient conditions in Proposition~\ref{inexact} are met. First, if we take $\delta=1/2$ and set $\alpha_1>1$ in Lemma \ref{isometryproperty}, we have
\begin{align*}
||\L_{T}^*\L_{T}-\I_T||\leq\frac{1}{2}
\end{align*}
when $N\geq C_{\alpha_1}\mu_{max}^2KJ\log^2(N)$ with probability at least $1-N^{-\alpha_1+1}$. Then, applying the same $\alpha_1$ in Lemma \ref{operatornorm} and setting $\gamma = \sqrt{M\log(MN/2)+\alpha_1 \log(N)}$, we have that $||\L||\leq \gamma$ with probability at least $1-N^{-\alpha_1}\geq 1-N^{-\alpha_1+1}$. In Theorem~\ref{gaussian}, we have proved that $\Y_T=\sign(\X_{0,T})$ and $||\Y_{T^C}||_{2,\infty}\leq\frac{1}{2}$ when $N\geq 40\alpha_2\mu_{max}^2KJ$ with probability at least $1-(M-J)e^{-\alpha_2}$ and $\alpha_2\geq\log(M-J)$.

Note that if $\alpha_2\geq(\alpha_1-1)\log(N)+\log(M-J)$, we have $(M-J)e^{-\alpha_2}\leq N^{-\alpha_1+1}$. Combining the above requirements on $N$, all conditions in Proposition~\ref{inexact} are satisfied with probability at least $1-3N^{-\alpha_1+1}$ when $N\geq \max\{C_{\alpha_1},40\}((\alpha_1-1)\log(N)+\log(M-J))\mu_{max}^2KJ\log^2(N)$. Furthermore,
\begin{equation*}
\begin{aligned}
&\quad \max\{C_{\alpha_1},40\}((\alpha_1-1)\log(N)+\log(M-J))\cdot\\
&\qquad\qquad\qquad\qquad\qquad\qquad\qquad\mu_{max}^2KJ\log^2(N)\\
&\leq C_{\alpha}(\log(N)+\log(M-J))\mu_{max}^2KJ\log^2(N)
\end{aligned}
\end{equation*}
if we set $C_{\alpha}=\max\{C_{\alpha_1},40\}\cdot\alpha_1$ and $\alpha = \alpha_1>1$, which yields the Theorem \ref{maintheom1} when $\A$ is a random Gaussian matrix.

\subsection{Constructing The Dual Certificate for The Fourier Case}

In this section, we construct a certificate $\Y$ that satisfies the inexact duality condition in Proposition \ref{inexact} when $\A$ is a random Fourier matrix. Specifically, we construct the dual certificate using the golfing scheme \cite{gross2011recovering} which has been widely applied in compressive sensing \cite{ahmed2014blind, candes2011probabilistic}. In the golfing scheme, a series of matrices in the range of $\L^*$ are constructed iteratively. In each iteration step, only some of the measurements are utilized to ensure independence between iterations. And the constructed matrices will converge to $\sign(\X_{0,T})$ on support $T$ while entries on $T^C$ are small. The goal is to find the conditions under which the final constructed matrix can serve as the certificate matrix.

According to Section (4.2.1) in  \cite{ling2015self}, there exists a partition of the $N$ measurements into $P$ disjoint subsets such that each subset, $\Gamma_p$, contains $Q$ elements and
\begin{align*}
\max_{1\leq p\leq P}||\B_p-\frac{Q}{N}\I_K||<\frac{Q}{4N},
\end{align*}
where $\B_p=\sum_{l\in\Gamma p}\b'_l\b_l'^H$ and $Q>C\mu_{max}^2K\log(N)$. So
\begin{align}
\max_{1\leq p\leq P}||\B_p||\leq\frac{5Q}{4N}.
\label{BL}
\end{align}
Define $\L_{p}(\X)=\{\b'^H_l\X \a'_{l}\}_{l\in\Gamma_p}$ and $0$ on entries $l\notin \Gamma_p$. $\L_{p}^*(\x)=\sum_{l\in\Gamma_p}x_l\b'_l\a'^H_{l}$. The golfing scheme iterates through
\begin{align}
\Y_p = \Y_{p-1}-\frac{N}{Q}\L_{p}^*\L_{p}(\Y_{p-1,T}-\text{sign}(\X_{0,T})), \quad \Y_0=0.
\label{golfing}
\end{align}
\begin{Theorem}
If $\X_0$ is the ground truth solution to problem \eqref{21}, there exists a matrix $\Y\in\L^*$ such that
\begin{align*}
||\Y_T-\sign(\X_{0,T})||_{F}\leq \frac{1}{4\sqrt2\gamma}\text{ and } ||\Y_{T^C}||_{2,\infty}\leq\frac{1}{2}
\end{align*}
with probability at least $1-2N^{-\alpha+1}$ for $\alpha>1$ when
\begin{align*}
&N=PQ,\qquad
P\geq \frac{\log(4\sqrt{2J}\gamma)}{\log2}
\end{align*}
and
\begin{align*}
Q\geq C_{\alpha}\mu_{max}^2KJ(\log(M-J)+\log(K+1)+\log(N))
\end{align*}
where $C_\alpha$ a constant determined by $\alpha$.
\label{finexact}
\end{Theorem}

\begin{proof}
If we define $\W_p=\Y_{p,T}-\text{sign}(\X_{0,T})$, (\ref{golfing}) gives
\begin{align}
\W_p = \frac{N}{Q}\left(\frac{Q}{N}-\L_{p,T}^*\L_{p,T}\right)(\W_{p-1}),
\label{Wdef}
\end{align}
where $\L_{p,T}(\X)=\{\b'^H_l\X \a'_{l,T}\}_{l\in\Gamma_p}$ with $0$ on entries $l\notin \Gamma_p$ and $\L_{p,T}^*(\x)=\sum_{l\in\Gamma_p}x_l\b'_l\a'^H_{l,T}$ which are used to generate the sequence $\Y_{p,T}$. And we can obtain
\begin{equation}
\begin{aligned}
||\W_p||_F &\leq ||\frac{N}{Q}(\frac{Q}{N}-\L_{p,T}^*\L_{p,T})||\cdot||\W_{p-1}||_F\leq \frac{1}{2}||\W_{p-1}||_F
\label{WF}
\end{aligned}
\end{equation}
with probability at least $1-N^{-\alpha+1}$ when $Q\geq C_{\alpha,1}\mu_{max}^2KJ\log(N)$ with $\alpha>1$ applying Lemma 4.6 in \cite{ling2015self}. Therefore,
\begin{align}
||\W_P||_F \leq 2^{-P}||\W_0||_F=2^{-P}||\text{sign}(\X_{0,T})||_F=2^{-P}\sqrt{J}.
\label{boundofw}
\end{align}
To ensure that $||\W_P||_F=||\Y_{P,T}-\sign(\X_{0,T})||_{F}\leq \frac{1}{4\sqrt2\gamma}$ where $\Y_{P}=\Y$ is the final constructed dual certificate after $P$ iterations, we need
\begin{align}
P\geq \frac{\log(4\sqrt{2J}\gamma)}{\log2}.
\label{pbound}
\end{align}

We now turn to find the conditions such that $||\Y_{T^C}||_{2,\infty}\leq\frac{1}{2}$. Note that substituting $\W_p$ into equation \eqref{golfing} yields
\begin{align*}
\Y = -\frac{N}{Q}\sum_{p=1}^P\L_{p}^*\L_{p}(\W_{p-1}).
\end{align*}
It is sufficient to show $||\Pi_{T^C}(\L_{p}^*\L_{p}\W_{p-1})||_{2,\infty}\leq2^{-p-1}\frac{Q}{N}$, where $\Pi_{T^C}$ is the projection operator which projects a matrix on the support $T^C$, to make $||\Y_{T^C}||_{2,\infty}\leq\frac{1}{2}$ because
\begin{equation*}
\begin{aligned}
&\quad||\Y_{T^C}||_{2,\infty}=||-\frac{N}{Q}\sum_{p=1}^P\Pi_{T^C}(\L_{p}^*\L_{p}(\W_{p-1}))||_{2,\infty}\\
&\leq \frac{N}{Q}\sum_{p=1}^P||\Pi_{T^C}(\L_{p}^*\L_{p}(\W_{p-1}))||_{2,\infty}\leq\frac{N}{Q}\sum_{p=1}^P\left(2^{-p-1}\frac{Q}{N}\right)\\
&=\sum_{p=1}^P2^{-p-1}=\frac{1}{2}(1-2^{-P})<\frac{1}{2}.\\
\end{aligned}
\end{equation*}
Defining $\PPhi_p$ to be $\PPhi$ with non-zero rows indexed by $\Gamma_p$ and zero otherwise, we have $\L_p(\X)=\PPhi_p\cdot\text{vec}(\X)$ and for a vector $\w= \text{vec}(\W)\in \C^{KM\times 1}$ where $\W\in\C^{K\times M}$ has support $T$,
\begin{equation}
\begin{aligned}
&||\Pi_{T^C}(\L_{p}^*\L_{p}(\W))||_{2,\infty}=\max_{i\in T^C}
\left\|
\begin{bmatrix}
\langle\PPhi_p^H\PPhi_p\w,\e_{K(i-1)+1}\rangle\\
\langle\PPhi_p^H\PPhi_p\w,\e_{K(i-1)+2}\rangle\\
\vdots\\
\langle\PPhi_p^H\PPhi_p\w,\e_{K(i-1)+K}\rangle\\
\end{bmatrix}
\right\|_2
\label{2infi}
\end{aligned}
\end{equation}
where $i$ is the column index and $\e_j$ is the $j$-th column of the identity matrix $\I_{KM}$. In addition,
\begin{equation*}
\begin{aligned}
\begin{bmatrix}
\langle\PPhi_p^H\PPhi_p\w,\e_{K(i-1)+1}\rangle\\
\langle\PPhi_p^H\PPhi_p\w,\e_{K(i-1)+2}\rangle\\
\vdots\\
\langle\PPhi_p^H\PPhi_p\w,\e_{K(i-1)+K}\rangle\\
\end{bmatrix}
&=
\sum_{l\in\Gamma_p}
\begin{bmatrix}
\langle\phi'_l\phi_l'^H\w,\e_{K(i-1)+1}\rangle\\
\langle\phi'_l\phi_l'^H\w,\e_{K(i-1)+2}\rangle\\
\vdots\\
\langle\phi'_l\phi_l'^H\w,\e_{K(i-1)+K}\rangle\\
\end{bmatrix}\\
&=\sum_{l\in\Gamma_p}\z_{l,i}.
\end{aligned}
\end{equation*}
Furthermore, we have $\E(\z_{l,i})=0$ because
\begin{equation*}
\begin{aligned}
\E(\z_{l,i})
&=\E\left(
\begin{bmatrix}
\langle\bar{\a}'_l\otimes \b'_l\cdot \bar{\a}'^H_l\otimes \b_l'^H\cdot\w,\e_{K(i-1)+1}\rangle\\
\langle\bar{\a}'_l\otimes \b'_l\cdot \bar{\a}'^H_l\otimes \b_l'^H\cdot\w,\e_{K(i-1)+2}\rangle\\
\vdots\\
\langle\bar{\a}'_l\otimes \b'_l\cdot \bar{\a}'^H_l\otimes \b_l'^H\cdot\w,\e_{K(i-1)+K}\rangle\\
\end{bmatrix}\right)\\
&=
\begin{bmatrix}
\langle(\I_M\otimes \b'_l\b_l'^H)\w,\e_{K(i-1)+1}\rangle\\
\langle(\I_M\otimes \b'_l\b_l'^H)\w,\e_{K(i-1)+2}\rangle\\
\vdots\\
\langle(\I_M\otimes \b'_l\b_l'^H)\w,\e_{K(i-1)+K}\rangle\\
\end{bmatrix}\\
&=
\begin{bmatrix}
\langle\text{vec}(\b'_l\b_l'^H\W),\e_{K(i-1)+1}\rangle\\
\langle\text{vec}(\b'_l\b_l'^H\W),\e_{K(i-1)+2}\rangle\\
\vdots\\
\langle\text{vec}(\b'_l\b_l'^H\W),\e_{K(i-1)+K}\rangle\\
\end{bmatrix}=\mathbf{0}
\end{aligned}
\end{equation*}
following $\E(\bar{\a}'_l\bar{\a}'^H_l)=\I_M$ since $\a'_l\in\C^{M\times 1}$ is the transpose of a random row of the $M\times M$ DFT matrix and $\b'_l\b_l'^H\W$ has support $T$ and $\mathbf{0}$ on $T^C$. Therefore, for $i\in T^C$, $\E(\z_{l,i})=0$.
Moreover,
\begin{align*}
||\z_{l,i}||_2\leq\sqrt{K\cdot\left(\frac{\mu_{max}^2\sqrt{KJ}||\w||_2}{N}\right)^2}=\frac{\mu_{max}^2K\sqrt{J}||\w||_2}{N}.
\end{align*}
Because each entry of $\z_{l,i}$ can be bounded by
\begin{equation*}
\begin{aligned}
&\quad|\langle\phi'_l\phi_l'^H\w,\e_{K(i-1)+j}\rangle|=|\e_{K(i-1)+j}^H\phi'_l\phi_l'^H\w|\\
&=||\e_{K(i-1)+j}^H\phi'_{l}||_2||(\bar{\a}'_{l}\otimes \b_l')^H\w||_2\\
&=||\e_{K(i-1)+j}^H\phi'_{l}||_2||(\bar{\a}'_{l,T}\otimes \b_l')^H\w||_2\\
&\leq \frac{\mu_{max}}{\sqrt{N}}||\bar{\a}'^H_{l,T}\otimes \b_l'^H||_2||\w||_2\leq\frac{\mu_{max}^2\sqrt{KJ}||\w||_2}{N}
\end{aligned}
\end{equation*}
where the third equality holds because $\w= \text{vec}(\W)$ and $\W$ has support T.
The variance of $\z_{l,i}$ is also bounded:
\begin{equation*}
\begin{aligned}
&\quad\max\left\{||\sum_{l\in\Gamma p}\E(\z_{l,i}\z_{l,i}^H)||,||\sum_{l\in\Gamma p}\E(\z_{l,i}^H\z_{l,i})||\right\}\\
&\leq \sum _{l\in\Gamma_p}\E (||\z_{l,i}||_2^2)\leq \frac{5\mu_{max}^2KQ||\w||^2_2}{4N^2}
\end{aligned}
\end{equation*}
because for each element of $||\z_{l,i}||_2^2$, we have
\begin{equation*}
\begin{aligned}
&\quad\E\left(|\langle\phi'_l\phi_l'^H\w,\e_{K(i-1)+j}\rangle|^2\right)=\E\left(||\e_{K(i-1)+i}^H\phi'_l\phi_l'^H\w||_2^2\right)\\
&\leq \frac{\mu_{max}^2}{N}\E\left(||\phi_l'^H\w||_2^2\right)=\frac{\mu_{max}^2}{N}\w^H(\I_M\otimes \b'_l\b_l'^H)\w
\end{aligned}
\end{equation*}
and therefore
\begin{equation*}
\begin{aligned}
\E (||\z_{l,i}||_2^2)\leq \frac{\mu_{max}^2K}{N}\w^H(\I_M\otimes \b'_l\b_l'^H)\w.
\end{aligned}
\end{equation*}
As a result,
\begin{equation}
\begin{aligned}
&\quad\sum_{l\in\Gamma_p}\E (||\z_{l,i}||_2^2)\leq\sum_{l\in\Gamma_p}\frac{\mu_{max}^2K}{N}\w^H(\I_M\otimes \b'_l\b_l'^H)\w\\
&=
\frac{\mu_{max}^2K}{N}\w^H(\I_M\otimes \B_p)\w\leq\frac{5\mu_{max}^2KQ||\w||^2_2}{4N^2}.
\label{var}
\end{aligned}
\end{equation}
The second inequality in (\ref{var}) applies the inequality (\ref{BL}) and $||\I_M\otimes \B_p||=||\I_M||\cdot||\B_p||$. We then apply the matrix Bernstein inequality from Theorem 1.6 in \cite{tropp2012user}. If we set $\w = \text{vec}(\W_{p-1})$ and we know from \eqref{boundofw} that $||\w||_2=||\W_{p-1}||_F\leq2^{-p+1}\sqrt{J}$, we obtain
\begin{equation*}
\begin{aligned}
\text{Pr}&\left(||\sum_{l\in\Gamma_p}\z_{l,i}||_2\geq t\right)\\
&\leq(K+1)\exp\left(\frac{-3t^2}{\frac{30\mu_{max}^2KQ||\w||_2^2}{4N^2}+\frac{2\mu_{max}^2K\sqrt{J}||\w||_2t}{N}}\right)\\
&\leq(K+1)\exp\left(\frac{-3Q}{128\mu_{max}^2KJ}\right)
\end{aligned}
\end{equation*}
where $t=2^{-p-1}\frac{Q}{N}$, for a particular $i\in T^C$ and $p$. We then take the union over all $i\in T^C$ and get
\begin{equation*}
\begin{aligned}
\text{Pr}&\left(||\Pi_{T^C}(\L_{p}^*\L_{p}(\W_{p-1}))||_{2,\infty}\geq 2^{-p-1}\frac{Q}{N}\right)\\
&\leq(M-J)(K+1)\exp\left(\frac{-3Q}{128\mu_{max}^2KJ}\right).
\end{aligned}
\end{equation*}
To ensure $||\Pi_{T^C}(\L_{p}^*\L_{p}(\W_{p-1}))||_{2,\infty}\leq2^{-p-1}\frac{Q}{N}$ for all $p$, we obtain
\begin{equation*}
\begin{aligned}
\text{Pr}&\left(||\Pi_{T^C}(\L_{p}^*\L_{p}(\W_{p-1}))||_{2,\infty}\leq 2^{-p-1}\frac{Q}{N},\quad \forall 1\leq p\leq P\right)\\
&>1-P(M-J)(K+1)\exp\left(\frac{-3Q}{128\mu_{max}^2KJ}\right)\\
&\geq 1-PN^{-\alpha}\geq 1-N^{-\alpha+1}
\end{aligned}
\end{equation*}
when $Q\geq \frac{128\mu_{max}^2KJ\alpha}{3}(\log(M-J)+\log(K+1)+\log(N))$ using the same $\alpha$ as in deriving equation (\ref{WF}). Setting $C_{\alpha}=\max\{C,C_{\alpha,1},\frac{128}{3}\alpha\}$, where $C$ is a constant comes from equation (\ref{BL}), gives us Theorem \ref{finexact}.
\end{proof}

\subsection{Proof of Theorem \ref{maintheom1} for Random Fourier Dictionary\label{fsum}}

We now complete the proof of Theorem \ref{maintheom1} in the case when $\A$ is a random Fourier matrix. First, combining the conditions and probabilities from Lemma \ref{isometryproperty} and \ref{operatornorm}, we know that the operator $\L$ satisfies the inequalities $||\L_{T}^*\L_{T}-\I_T||\leq\frac{1}{2}$ and $||\L||\leq \gamma=\sqrt{2M\log(2KM)+2M+1}$ with probability at least $1-(N+1)N^{-\alpha}\geq1-2N^{-\alpha+1}$ when $N\geq C_{\alpha,1}\mu_{max}^2KJ\log(N)$ for some constant, $C_{\alpha,1}$, that grows linearly with $\alpha>1$.

Applying the same $\alpha$ in Theorem \ref{finexact}, the desired dual matrix exists with probability at least $1-2N^{-\alpha+1}$ when $N\geq C_{\alpha,2}\mu_{max}^2KJ\log(4\sqrt{2J}\gamma)(\log(M-J)+\log(K+1)+\log(N))$. Merging the requirement on $N$ by setting $C_{\alpha}=\max\{C_{\alpha,1},C_{\alpha,2}\}$ and combining the probabilities, we complete the proof by applying Proposition \ref{inexact}.

\section{Proof of Theorem \ref{maintheom2}}
\label{sec:proof2}

To derive our recovery guarantee in the presence of measurement noise, the main ingredient of the proof is Theorem \ref{noisyinexact} which is a variation of the Theorem 4.33 in \cite{foucart2013mathematical} from the infinity norm optimization to $\ell_{2,1}$ norm optimization problem.

\begin{Theorem}
Define $\PPhi\in\C^{N\times KM}$ and $\PPhi\cdot\text{vec}(\X)=\L(\X)$. Suppose the ground truth $\X_0$ to \eqref{noisy} has $J$ non-zero columns with support $T$ and the measurement vector $\y=\L(\X_0)+\n$ with $||\n||_2\leq \eta$. For $\delta,\beta,\theta,\gamma,\tau>0$ and $\delta<1$, assume that
\begin{align*}
\max_{
\substack{i\in T^C }
}
&||\PPhi_T^H[\PPhi_{K(i-1)+1}\cdots\quad \PPhi_{K(i-1)+K}]||\leq\beta,
\\
&||\PPhi_T^H\PPhi_T-\I_T||\leq\delta
\end{align*}
and that there exists a matrix $\Y=\L^*(\p)\in \C^{K\times M}$ such that
\begin{equation*}
\begin{aligned}
||\Y_T-\sign(&\X_{0,T})||_F\leq\frac{1}{4\sqrt{2}\gamma},\quad ||\Y_{T^C}||_{2,\infty}\leq\theta,
\\
&\text{and}\quad ||\p||_{2}\leq\tau\sqrt{J}.
\end{aligned}
\end{equation*}
If $\rho:=\theta+\frac{\beta}{4\sqrt{2}\gamma(1-\delta)}<1$, then the minimizer, $\hat{\X}$, to \eqref{noisy} satisfies
\begin{equation*}
\begin{aligned}
||\hat{\X}&-\X_0||_F\leq\left(C_1+C_2\sqrt{J}\right)\eta
\end{aligned}
\end{equation*}
where $C_1$ and $C_2$ are two constants depending on $\delta,\beta,\theta,\gamma,\tau$.

\label{noisyinexact}
\end{Theorem}

\begin{proof}
Due to our assumption on the noise, $\X_0$ is a feasible solution. Assume the final minimizer to \eqref{noisy} is $\hat{\X}=\X_0+\H$, which implies
\begin{equation*}
\begin{aligned}
||\X_0||_{2,1}&\geq ||\X_0+\H||_{2,1}=||\X_{0,T}+\H_T||_{2,1}+||\H_{T^C}||_{2,1}\\
               &\geq |\langle\X_{0,T}+\H_T,\sign(\X_{0,T})\rangle|+||\H_{T^C}||_{2,1}\\
               &\geq||\X_{0}||_{2,1}-|\langle\H_T,\sign(\X_{0,T})\rangle|+||\H_{T^C}||_{2,1}
\end{aligned}
\end{equation*}
where the second inequality comes from equation (\ref{ineq}). Thus
\begin{equation}
\begin{aligned}
||\H_{T^C}||_{2,1}&\leq |\langle\H_T,\sign(\X_{0,T})\rangle|\\
                  &\leq |\langle\H_T,\sign(\X_{0,T})-\Y_{T}\rangle|+|\langle\H_T,\Y_{T}\rangle|\\
                  &\leq \frac{1}{4\sqrt{2}\gamma}||\H_T||_F+|\langle\H,\Y\rangle|+|\langle\H_{T^C},\Y_{T^C}\rangle|\\
                  &\leq \frac{1}{4\sqrt{2}\gamma}||\H_T||_F + 2\tau\eta\sqrt{J}+\theta||\H_{T^C}||_{2,1}.
\label{HTC21}
\end{aligned}
\end{equation}
The last inequality comes from the H\"{o}lder inequality and our assumption $||\n||\leq \eta$, which tells us
\begin{equation*}
\begin{aligned}
||\L(\H)||_2&=||\L(\hat{\X}-\X_0)||_2=||\L(\hat{\X})-\L(\X_0)||_2\\
&\leq||\L(\hat{\X})-\y||_2+||\y-\L(\X_0)||_2\leq2\eta
\end{aligned}
\end{equation*}
and
\begin{equation*}
\begin{aligned}
|\langle\H,\Y\rangle|&=|\langle\H,\L^*(\p)\rangle|=|\langle\L(\H),\p\rangle|\leq \tau\sqrt{J}||\L(\H)||_2\\
&\leq 2\tau\eta\sqrt{J}.
\end{aligned}
\end{equation*}
Moreover, $||\H_T||_F$ can also be bounded as follows.
\begin{equation}
\begin{aligned}
&\quad ||\H_T||_F = ||(\tilde{\PPhi}_T^H\tilde{\PPhi}_T)^{-1}\tilde{\PPhi}_T^H\tilde{\PPhi}_T\cdot\text{vec}(\tilde{\H}_T)||_2\\
           &\leq\frac{1}{1-\delta}||\tilde{\PPhi}_T^H\tilde{\PPhi}_T\cdot\text{vec}(\tilde{\H}_T)||_2=\frac{1}{1-\delta}||\PPhi_T^H\PPhi_T\cdot\text{vec}(\H_T)||_2\\
           &=\frac{1}{1-\delta}||\PPhi_T^H(\PPhi\cdot\text{vec}(\H)-\PPhi_{T^C}\cdot\text{vec}(\H_{T^C}))||_2\\
           &\leq  \frac{1}{1-\delta}||\PPhi_T^H\PPhi\cdot\text{vec}(\H)||_2 +\frac{1}{1-\delta}||\PPhi_T^H\PPhi_{T^C}\cdot\text{vec}(\H_{T^C})||_2\\
&= \frac{1}{1-\delta}||\PPhi_T^H\L(\H)||_2 +\frac{1}{1-\delta}||\PPhi_T^H\PPhi_{T^C}\cdot\text{vec}(\H_{T^C})||_2\\
&\leq \frac{2\eta\sqrt{1+\delta}}{1-\delta}+\frac{\beta}{1-\delta}||\H_{T^C}||_{2,1}
\label{HTF}
\end{aligned}
\end{equation}
because $||\PPhi_T^H\PPhi_T-\I_T||\leq\delta$ ensures that $||(\tilde{\PPhi}_T^H\tilde{\PPhi}_T)^{-1}||\leq\frac{1}{1-\delta}$ and $||\PPhi_T^H||\leq\sqrt{1+\delta}$ according to Lemma A.12 and Proposition A.15 in \cite{foucart2013mathematical} respectively. Furthermore,
\begin{equation*}
\begin{aligned}
&\qquad||\PPhi_T^H\PPhi_{T^C}\cdot\text{vec}(\H_{T^C})||_2\\
&= ||\sum_{i\in T^C}\PPhi_T^H[\PPhi_{K(i-1)+1}\cdots\quad \PPhi_{K(i-1)+K}]\h_i||_2\\
&\leq
\sum_{i\in T^C}||\PPhi_T^H[\PPhi_{K(i-1)+1}\cdots\quad \PPhi_{K(i-1)+K}]||\cdot||\h_i||_2\\
&\leq
\sum_{i\in T^C}\beta||\h_i||_2=\beta||\H_{T^C}||_{2,1}
\end{aligned}
\end{equation*}
in which $\h_i$ is the $i$-th column of $\H$. By setting $\rho=\theta+\frac{\beta}{4\sqrt{2}\gamma(1-\delta)}$, $\mu=\frac{\sqrt{1+\delta}}{1-\delta}$ and substituting the inequality \eqref{HTF} into \eqref{HTC21}, we obtain
\begin{align}
||\H_{T^C}||_{2,1}\leq\frac{\eta\mu}{2\sqrt{2}\gamma(1-\rho)}+\frac{2\tau\eta\sqrt{J}}{1-\rho}.
\label{HTC21COM}
\end{align}
Substituting inequality \eqref{HTC21COM} into \eqref{HTF} yields
\begin{align*}
||\H_{T}||_{F}\leq2\eta\mu+\frac{\beta}{1-\delta}\left(\frac{\eta\mu}{2\sqrt{2}\gamma(1-\rho)}+\frac{2\tau\eta\sqrt{J}}{1-\rho}\right).
\end{align*}
Combining the above two inequalities, we obtain
\begin{equation}
\begin{aligned}
||\H||_F&\leq||\H_T||_F+||\H_{T^C}||_F\leq ||\H_T||_F+||\H_{T^C}||_{2,1}\\
&\leq\bigg(2\mu+\frac{\mu}{2\sqrt{2}\gamma(1-\rho)}+\frac{\beta\mu}{2\sqrt{2}\gamma(1-\delta)(1-\rho)}\\
&\qquad+\left(\frac{2\tau}{1-\rho}+\frac{2\beta\tau}{(1-\delta)(1-\rho)} \right)\sqrt{J}\bigg)\eta\\
&=\left(C_1+C_2\sqrt{J}\right)\eta.
\label{noisyC}
\end{aligned}
\end{equation}
\end{proof}

Next, we specify the values of the variables $\theta$, $\tau$, $\delta$ and $\beta$ when $\A$ is a random Gaussian and Fourier matrix. The Orlicz-1 norm \cite{ahmed2014blind} and associated matrix Bernstein inequality are needed for determining the value of $\beta$ when $\A$ is Gaussian. Specifically, the Orlicz-1 norm is defined as \cite{ahmed2014blind}
\begin{align}
||\Z||_{\psi_1}=\inf_{u\geq 0}\{\E[\exp(||\Z||/u)]\leq2\}.
\label{ol1}
\end{align}
Its associated matrix Bernstein inequality is provided in Proposition 3 in \cite{ahmed2014blind} which can be rewritten as
\begin{Proposition}
Let $\Z_1,...,\Z_N$ be independent $M\times M$ random matrices with $\E(\Z_j)=0$. Suppose
\begin{align*}
\max_{1\leq j\leq N}||\Z_j||_{\psi_1}\leq R
\end{align*}
and define
\begin{align*}
\sigma^2=\max\Bigg\{||\sum_{j=1}^N\E(\Z_j\Z_j^H)||,||\sum_{j=1}^N\E(\Z_j^H \Z_j)||\Bigg\}.
\end{align*}
Then there exists a constant $C$ such that for $t > 0$
\begin{align*}
\text{Pr}\left(||\sum_{j=1}^N\Z_j||>t\right)\leq 2M\exp\left(-\frac{1}{C}\frac{t^2}{\sigma^2+\log\left(\frac{\sqrt{N}R}{\sigma}\right)Rt}\right).
\end{align*}
\label{orlicz1b}
\end{Proposition}
The following theorem utilizes the Proposition \ref{orlicz1b} and depicts the conditions under which $\beta=1$.

\begin{Theorem}
For $\PPhi$ defined in \eqref{phi} and $\L(\X)=\PPhi\text{vec}(\X)$,
\begin{align*}
\max_{
\substack{i\in T^C }
}
&||\PPhi_T^H[\PPhi_{K(i-1)+1}\cdots\quad \PPhi_{K(i-1)+K}]||\leq1
\end{align*}
with probability at least $1-N^{-\alpha+1}$
\begin{itemize}
\item if $\A$ is a random Gaussian matrix and
\begin{equation*}
\begin{aligned}
N&\geq C_{\alpha}\mu_{max}^2KJ\left(\log(C\mu_{max}\sqrt{KJ})C+1\right)\cdot\\
&\qquad\qquad\qquad\left(\log(KM)+\log(M-J)+\log(N)\right),
\end{aligned}
\end{equation*}
\item if $\A$ is a random Fourier matrix and
\begin{align*}
N\geq C_{\alpha}\mu_{max}^2KJ\left(\log(KM)+\log(M-J)+\log(N)\right),
\end{align*}
\end{itemize}
where $C_{\alpha}$ is a constant that grows linearly with $\alpha>1$ and $C$ is a constant.
\label{beta}
\end{Theorem}

\begin{proof}
We first prove the Gaussian case; the Fourier case is very similar. Note that for an arbitrary $i\in T^C$
\begin{equation*}
\begin{aligned}
&\quad||\PPhi_T^H[\PPhi_{K(i-1)+1}\cdots\quad \PPhi_{K(i-1)+K}]||\\
&=||\PPhi_T^H\PPhi_i||=||\sum_{j=1}^N \left(\bar{\a}'_{j,T}\otimes \b'_{j}\right)\cdot\left(\bar{\a}'^H_{j,i}\otimes \b'^H_j\right)||\\
&= ||\sum_{j=1}^N\left(\bar{\a}'_{j,T}\bar{\a}_{j,i}'^H)\otimes (\b'_j\b_j'^H\right)||= ||\sum_{j=1}^N \Z_j||
\end{aligned}
\end{equation*}
where $\PPhi_i\in \C^{N\times KM}$ is $\PPhi$ but only contains values in the $\left(K(i-1)+1\right)$-th to $\left(K(i-1)+K\right)$-th columns and is zero otherwise. $\PPhi_i$ can also be viewed as an extension of $[\PPhi_{T^C,K(i-1)+1}\cdots\quad \PPhi_{T^C,K(i-1)+K}]$ by padding zero columns. Moreover, $\bar{\a}_{j,i}'$ is the conjugate of the $j$-th column of $\A^T$ who has only one non-zero value in the $i$-th entry. In addition, $\E(\Z_j)=\E(\bar{\a}'_{j,T}\bar{\a}_{j,i}'^H\otimes \b'_j\b_j'^H)=\E(\bar{\a}'_{j,T}\bar{\a}_{j,i}'^H)\otimes \b'_j\b_j'^H=\mathbf{0}$ for $i\in T^C$. By applying the property of the Kronecker product, we estimate the spectral norm of $\Z_j$ which can be used to determine its Orlicz-1 norm:
\begin{equation*}
\begin{aligned}
||\Z_j||&=||\bar{\a}'_{j,T}\bar{\a}_{j,i}'^H\otimes \b'_j\b_j'^H||=||\b'_j\b_j'^H||\cdot||\bar{\a}'_{j,T}\bar{\a}_{j,i}'^H||\\
&=|\b_j'^H \b'_j|\cdot||\bar{\a}'_{j,T}\bar{\a}_{j,i}'^H||\leq \frac{\mu_{max}^2K}{N}||\bar{\a}'_{j,T}\bar{\a}_{j,i}'^H||\\
&=\frac{\mu_{max}^2K}{N}||\bar{\a}'_{j,T}||_2||\bar{\a}_{j,i}'||_2\\
&\leq \frac{\mu_{max}^2K}{N}\cdot\frac{||\bar{\a}'_{j,T}||_2^2+||\bar{\a}_{j,i}'||_2^2}{2}=\frac{\mu_{max}^2K}{2N}||\bar{\a}'_{j,\{T,i\}}||_2^2
\end{aligned}
\end{equation*}
in which $\bar{\a}'_{j,\{T,i\}}$ contains non-zero values in the entries indexed by $\{T,i\}$. Therefore, $||\bar{\a}'_{j,\{T,i\}}||_2^2$ follows the Chi-squared distribution with $J+1$ degrees of freedom which implies that $||\Z_j||_{\psi_1}\leq \frac{C\mu_{max}^2K(J+1)}{2N}\leq \frac{C\mu_{max}^2K\cdot2J}{2N}=\frac{C\mu_{max}^2KJ}{N}=R$ for some constant $C$ according to the proof of Lemma 4.7 in \cite{ling2015self} and the definition of Orlicz-1 norm in \eqref{ol1}. Moreover,
\begin{equation*}
\begin{aligned}
&\quad ||\sum_{j=1}^N\E(\Z_j^H\Z_j)||\\
&=||\sum_{j=1}^N\E\left[(\bar{\a}'_{j,i}\bar{\a}_{j,T}'^H)\otimes (\b'_j\b_j'^H)\cdot(\bar{\a}'_{j,T}\bar{\a}_{j,i}'^H)\otimes (\b'_j\b_j'^H)\right]||\\
&=||\sum_{j=1}^N\E\left(\bar{\a}'_{j,i}\bar{\a}_{j,T}'^H\bar{\a}'_{j,T}\bar{\a}_{j,i}'^H\right)\otimes (\b'_j\b_j'^H \b'_j\b_j'^H)||\\
&=||J\I_{M,i}\otimes \left(\sum_{j=1}^N||\b_j'||_2^2\cdot \b'_j\b_j'^H\right)||\\
&\leq \frac{\mu_{max}^2KJ}{N}||\I_{M,i}||\cdot||\sum_{j=1}^N\b_j'\b_j'^H||=\frac{\mu_{max}^2KJ}{N}
\end{aligned}
\end{equation*}
following from the fact that $\E\left(\bar{\a}'_{j,i}\bar{\a}_{j,T}'^H\bar{\a}'_{j,T}\bar{\a}_{j,i}'^H\right)=J\I_{M,i}$ for all $j$ and $\sum_{j=1}^N\b_j'\b_j'^H=\I_K$ from the assumption. On the other hand,
\begin{equation*}
\begin{aligned}
&\quad ||\sum_{j=1}^N\E(\Z_j\Z_j^H)||\\
&=||\sum_{j=1}^N\E\left(\bar{\a}'_{j,T}\bar{\a}_{j,i}'^H\bar{\a}'_{j,i}\bar{\a}_{j,T}'^H\right)\otimes \b'_j\b_j'^H \b'_j\b_j'^H||\\
&=||\I_{M,T}\otimes \left(\sum_{j=1}^N||\b_j'||_2^2\cdot \b'_j\b_j'^H\right)||\\
&\leq \frac{\mu_{max}^2K}{N}||\I_{M,T}||\cdot||\sum_{j=1}^N\b_j'\b_j'^H||=\frac{\mu_{max}^2K}{N}.
\end{aligned}
\end{equation*}
Therefore, $\max\{||\sum_{j=1}^N\E(\Z_j\Z_j^H)||,||\sum_{j=1}^N\E(\Z_j^H \Z_j)||\}=\frac{\mu_{max}^2KJ}{N}=\sigma^2$.
Substituting the variables $R$ and $\sigma^2$ into Proposition \ref{orlicz1b} and taking the union bound over all $i\in T^C$ results in
\begin{equation*}
\begin{aligned}
&\text{Pr}\left(\max_{i\in T^C}||\PPhi_T^H[\PPhi_{T^C,K(i-1)+1}\cdots\quad \PPhi_{T^C,K(i-1)+K}]||>1\right)\\
&\leq2(M-J)KM\exp\Bigg(-\frac{1}{C_{0}}\cdot\\
&\qquad\qquad\qquad\frac{N}{\mu_{max}^2KJ+\log\left(C\mu_{max}\sqrt{KJ}\right)C\mu_{max}^2 KJ}\Bigg).
\end{aligned}
\end{equation*}
Define a variable $\alpha>1$ and set
\begin{equation*}
\begin{aligned}
N&\geq C_{\alpha}\mu_{max}^2KJ\left(\log(C\mu_{max}\sqrt{KJ})C+1\right)\cdot\\
&\qquad\qquad\qquad\left(\log(KM)+\log(M-J)+\log(N)\right)\\
&\geq C_0\mu_{max}^2KJ\left(\log(C\mu_{max}\sqrt{KJ})C+1\right)\cdot\\
&\qquad\qquad\qquad\left(\log(KM)+\log(M-J)+\alpha\log(N)\right),
\end{aligned}
\end{equation*}
where $C_{\alpha}=C_0\alpha$. Simplifying the probability term gives
\begin{equation*}
\begin{aligned}
&\text{Pr}\left(\max_{i\in T^C}||\PPhi_T^H[\PPhi_{T^C,K(i-1)+1}\cdots\quad \PPhi_{T^C,K(i-1)+K}]||\leq1\right)\\
&> 1-2N^{-\alpha}\geq 1-N\cdot N^{-\alpha}=1- N^{-\alpha+1}.
\end{aligned}
\end{equation*}

Following the same procedures, when $\A$ is a random Fourier matrix and for any $i\in T^C$, we have $\E(\Z_j)=\E(\bar{\a}'_{j,T}\bar{\a}_{j,i}'^H)\otimes \b'_j\b_j'^H=0$, $||\Z_j||=\frac{\mu_{max}^2K}{N}||\bar{\a}'_{j,T}||_2||\bar{\a}_{j,i}'||_2=\frac{\mu_{max}^2K\sqrt{J}}{N}=R$ and $\sigma^2=\frac{\mu_{max}^2KJ}{N}$. The matrix Bernstein inequality implies
\begin{equation*}
\begin{aligned}
&\text{Pr}\left(\max_{i\in T^C}||\PPhi_T^H[\PPhi_{T^C,K(i-1)+1}\cdots\quad \PPhi_{T^C,K(i-1)+K}]||>1\right)\\
&\leq2(M-J)KM\exp\left(-\frac{N}{2\mu_{max}^2KJ+2/3\mu_{max}^2K\sqrt{J}}\right).
\end{aligned}
\end{equation*}
Similarly, if we define a variable $\alpha>1$ and let
\begin{equation*}
\begin{aligned}
N&\geq C_{\alpha}\mu_{max}^2KJ\left(\log(KM)+\log(M-J)+\log(N)\right)\\
&\geq(2\mu_{max}^2KJ+ \frac{2}{3}\mu_{max}^2K\sqrt{J})(\log(KM)+\log(M-J)\\
&\qquad\qquad\qquad+\alpha\log(N)),
\end{aligned}
\end{equation*}
by setting $C_{\alpha}=\frac{8}{3}\alpha$, simplifying the probability gives us
\begin{equation*}
\begin{aligned}
&\text{Pr}\left(\max_{i\in T^C}||\PPhi_T^H[\PPhi_{T^C,K(i-1)+1}\cdots\quad \PPhi_{T^C,K(i-1)+K}]||\leq1\right)\\
&> 1-2N^{-\alpha}\geq 1-N\cdot N^{-\alpha}=1- N^{-\alpha+1}.
\end{aligned}
\end{equation*}
\end{proof}

\subsection{Proof of Theorem \ref{maintheom2} for Random Gaussian Dictionary\label{noisyC1}}
According to Section \ref{gsum}, $||\PPhi_T^H\PPhi_T-\I_T||\leq\frac{1}{2}=\delta$, $||\Y_{T^C}||_{2,\infty}\leq\frac{1}{2}=\theta$ and $\gamma = \sqrt{M\log(MN/2)+\alpha \log(N)}$ with probability at least $1-3N^{-\alpha+1}$ when $\frac{N}{\log^2N}\geq C_{\alpha,1}\mu_{max}^2KJ(\log(N)+\log(M-J))$. Moreover, in Theorem \ref{gaussian}, where we construct the dual certificate matrix when $\A$ is a random Gaussian matrix, we define $\p = \tilde{\PPhi}_T(\tilde{\PPhi}_T^H\tilde{\PPhi}_T)^{-1}\text{vec}(\sign(\tilde{\X}_{0,T}))\in\C^{N\times 1}$ and $||\PPhi_T^H\PPhi_T-\I_T||\leq\frac{1}{2}$ leads to $||(\tilde{\PPhi}_T^H\tilde{\PPhi}_T)^{-1}||\leq 2$. So
\begin{equation*}
\begin{aligned}
||\p||_2&=\sqrt{\text{vec}(\sign(\tilde{\X}_{0,T}))^H(\tilde{\PPhi}_T^H\tilde{\PPhi}_T)^{-1}\text{vec}(\sign(\tilde{\X}_{0,T}))}\\
&\leq \sqrt{2||\text{vec}(\sign(\tilde{\X}_{0,T}))||_2^2}=\sqrt{2J}
\end{aligned}
\end{equation*}
which implies $\tau=\sqrt{2}$. If we use the same $\alpha$ in Theorem \ref{beta}, we have $\beta=1$ with probability at least $1-N^{-\alpha+1}$ when
\begin{equation*}
\begin{aligned}
N&\geq C_{\alpha,2}\mu_{max}^2KJ\left(\log(C\mu_{max}\sqrt{KJ})C+1\right)\cdot\\
&\qquad\qquad\qquad\left(\log(MK)+\log(M-J)+\log(N)\right).
\end{aligned}
\end{equation*}
Combining the requirement on $N$ and setting $C_{\alpha}=\max\{C_{\alpha,1},C_{\alpha,2}\}$ yield
\begin{equation}
\begin{aligned}
\frac{N}{\log^2N}\geq &C_{\alpha}\mu_{max}^2KJ\left(\log(C\mu_{max}\sqrt{KJ})C+1\right)\cdot\\
&\qquad\qquad(\log(M-J)+\log(MK)+\log(N)).
\label{gaussiannoisy}
\end{aligned}
\end{equation}
Therefore, the conditions in Theorem \ref{noisyinexact} are satisfied with probability at least $1-4N^{\alpha+1}$ when $N$ is as defined in equation (\ref{gaussiannoisy}). In addition, after substituting the parameters $\rho= \theta+\frac{\beta}{4\sqrt{2}\gamma(1-\delta)}=\frac{1}{2}+\frac{1}{2\sqrt{2}\gamma}<1$ and $\mu=\frac{\sqrt{1+\delta}}{1-\delta}=\sqrt{6}$ into \eqref{noisyC}, $2\mu+\frac{\mu}{2\sqrt{2}\gamma(1-\rho)}+\frac{\beta\mu}{2\sqrt{2}\gamma(1-\delta)(1-\rho)}=2\sqrt{6}+\frac{3\sqrt{6}}{\sqrt{2}\gamma-1}\leq5\sqrt{6}=C_1$ and $\frac{2\tau}{1-\rho}+\frac{2\beta\tau}{(1-\delta)(1-\rho)}=\frac{24\gamma}{\sqrt{2}\gamma-1}\leq24=C_2$.

\subsection{Proof of Theorem \ref{maintheom2} for Random Fourier Dictionary\label{noisyC2}}
In the proof of Theorem \ref{finexact}, we have derived $\Y = -\frac{N}{Q}\sum_{p=1}^P\L_{p}^*\L_{p}(\W_{p-1})$. Since the sets $\Gamma_p$ are disjoint, the indices of non-zero entries of $\L_{p}(\W_{p-1})$ for different $p$ are disjoint and $\Y = \L^*(-\frac{N}{Q}\sum_{p=1}^P\L_{p}(\W_{p-1}))=\L^*(\p)$. Moreover, $\W_{p-1}$ has support $T$ from its definition in \eqref{Wdef} which gives us
\begin{equation*}
\begin{aligned}
&\quad||\p||_2^2\leq\frac{N^2}{Q^2}\sum_{p=1}^P||\L_{p}(\W_{p-1})||_2^2
=\frac{N^2}{Q^2}\sum_{p=1}^P||\L_{p,T}(\W_{p-1})||_2^2\\
&= \frac{N^2}{Q^2}\sum_{p=1}^P \text{vec}(\W_{p-1})^H\PPhi_{p,T}^H\PPhi_{p,T}\text{vec}(\W_{p-1})\\
&\leq \frac{N^2}{Q^2}\sum_{p=1}^P||\PPhi_{p,T}^H\PPhi_{p,T}||\cdot||\W_{p-1}||_F^2\leq \frac{N^2}{Q^2}\sum_{p=1}^P \frac{3Q}{2N}4^{-p+1}J\\
&\leq\frac{2NJ}{Q}=2PJ
\end{aligned}
\end{equation*}
because $||\PPhi_{p,T}^H\PPhi_{p,T}||\leq \frac{3Q}{2N}$ and $||\W_{p-1}||_F^2\leq 4^{-p+1}J$ following from Lemma 4.6 in \cite{ling2015self} and equation (\ref{boundofw}) respectively. $\PPhi_{p,T}$ is $\PPhi$ constructed with $\A_T$ and only rows indexed by $\Gamma_p$ are non-zero. Therefore, $||\p||_2\leq\sqrt{2PJ}$ and $\tau=\sqrt{2P}$ with $P\geq \log(4\sqrt{2J}\gamma)/\log2$ defined in equation (\ref{pbound}). In addition, from Section \ref{fsum} and Theorem \ref{maintheom1}, we have $\delta=\frac{1}{2}$, $\theta=\frac{1}{2}$ and $\gamma = \sqrt{2M\log(2KM)+2M+1}$ with probability at least $1-4N^{-\alpha+1}$ when
\begin{equation*}
\begin{aligned}
N\geq &C_{\alpha ,1}\mu_{max}^2KJ\log(4\sqrt{2J}\gamma)\cdot\\
&\qquad\qquad\qquad(\log(M-J)+\log(K+1)+\log(N)).
\end{aligned}
\end{equation*}
Applying the same $\alpha$ to Theorem \ref{beta}, $\beta=1$ with probability at least $1-N^{-\alpha+1}$ when $N\geq C_{\alpha,2}\mu_{max}^2KJ\left(\log(KM)+\log(M-J)+\log(N)\right)$.
One can easily examine that $\rho= \theta+\frac{\beta}{4\sqrt{2}\gamma(1-\delta)}=\frac{1}{2}+\frac{1}{2\sqrt{2}\gamma}<1$.

If we set $C_\alpha = \max\{C_{\alpha,1},C_{\alpha,2}\}$ and merge the requirements on $N$, we obtain
\begin{equation}
\begin{aligned}
N\geq &C_{\alpha}\mu_{max}^2KJ\log(4\sqrt{2J}\gamma)\cdot\\
&\qquad\qquad(\log(M-J)+\log(MK)+\log(N)).
\label{fouriernoisy}
\end{aligned}
\end{equation}
Thus, the conditions in Theorem \ref{noisyinexact} are satisfied with probability at least $1-5N^{-\alpha+1}$ when $N$ satisfies \eqref{fouriernoisy}. Moreover, since $\mu=\frac{\sqrt{1+\delta}}{1-\delta}$, $2\mu+\frac{\mu}{2\sqrt{2}\gamma(1-\rho)}+\frac{\beta\mu}{2\sqrt{2}\gamma(1-\delta)(1-\rho)}=2\sqrt{6}+\frac{3\sqrt{6}}{\sqrt{2}\gamma-1}\leq5\sqrt{6}=C_1$ and $\frac{2\tau}{1-\rho}+\frac{2\beta\tau}{(1-\delta)(1-\rho)}=\frac{24\gamma\sqrt{P}}{\sqrt{2}\gamma-1}\leq24\sqrt{P}=C_2\sqrt{P}$ with $P\geq \log(4\sqrt{2J}\gamma)/\log2$.

\section{Numerical Simulations}
\label{sec:sims}

Here we present numerical simulations that illustrate and support our theoretical results. We set $\B\in \C^{N\times K}$ to be the first $K$ columns of the normalized DFT matrix $\frac{1}{\sqrt{N}}\F\in \C^{N\times N}$. The parameters $c_j$ and $\h_j$ are generated by sampling independently from the standard normal distribution, and the $J$ non-zero columns of the ground truth solution $\X_0=[c_j\h_j \cdots c_M\h_M]$ are selected uniformly. 40 simulations are run for each setting, based on which we compute the percentage of successful recovery. Both the dictionary, $\A$, and the ground truth solution, $\X_0$, including the support and its content, are sampled independently for each simulation. We solve problems \eqref{21} and \eqref{noisy} via CVX \cite{grant2008cvx}, and in the noiseless case if the relative error between the solution $\hat{\X}$ and the ground truth $\X_0$ is smaller than $10^{-5}$, $\frac{||\hat{\X}-\X_0||_F}{||\X_0||_F}\leq10^{-5}$, we count it as a successful recovery.

\begin{figure}[h]
\begin{center}
   \includegraphics[width=0.8\linewidth]{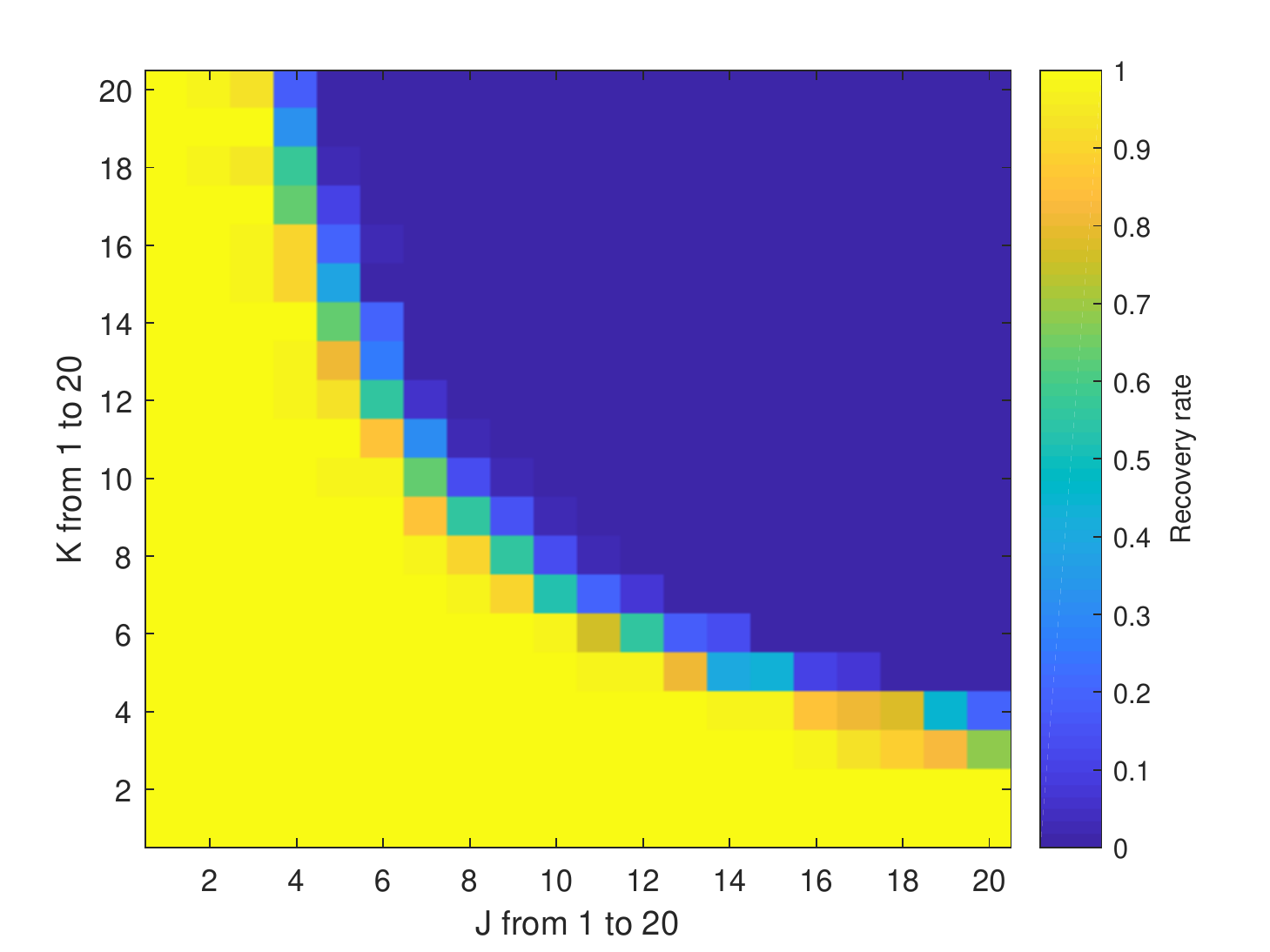}
\end{center}
   \caption{ The relation between the subspace dimension of the sensing matrix, $K$, and the number of committed atoms, $J$, in terms of the success recovery rate when $\A$ is a random Gaussian matrix. }
\label{fig:KJgaussian}
\end{figure}

\begin{figure}[h]
\begin{center}
   \includegraphics[width=0.8\linewidth]{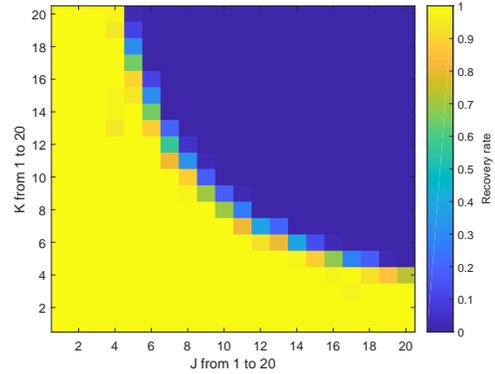}
\end{center}
   \caption{The relation between the subspace dimension of the sensing matrix, $K$, and the number of committed atoms, $J$, in terms of the success recovery rate when $\A$ is a random Fourier matrix. }
\label{fig:KJfourier}
\end{figure}

\subsection{ The Sufficient Number of Measurement\label{sufficientN}}
In the first noiseless simulation, we examine the recovery rate with respect to the parameters $K$ and $J$. We fix $M=200$ and $N=100$ and let $K$ and $J$ range from $1$ to $20$. The results are summarized in the phase transition plots of Fig. \ref{fig:KJgaussian} for the random Gaussian dictionary and Fig. \ref{fig:KJfourier} for the random Fourier dictionary. The results for the two dictionaries are similar. The reciprocal nature of the phase transition boundary supports the linear scaling with $KJ$ in equations (\ref{maintheom11}) and (\ref{maintheom12}). Roughly when $KJ\leq60$, the recovery success rate is satisfactory.

To further illustrate the linear scaling of the required number of measurements $N$ with respect to $K$ and $J$, we fix $M=200$ and $K=5$, and let $N$ and $J$ range from $30$ to $100$ and $1$ to $20$, respectively. The results are recorded in Figs. \ref{fig:KJsepGaussianJ} and \ref{fig:KJsepFourierJ} for the random Gaussian and Fourier dictionaries, respectively. The same simulation but switching the roles of $K$ and $J$ is also implemented, and the results are shown in Figs. \ref{fig:KJsepGaussianK} and \ref{fig:KJsepFourierK}. These results support the linear scaling of Theorem~\ref{maintheom1}.

\begin{figure}[t]
\begin{center}
   \includegraphics[width=0.8\linewidth]{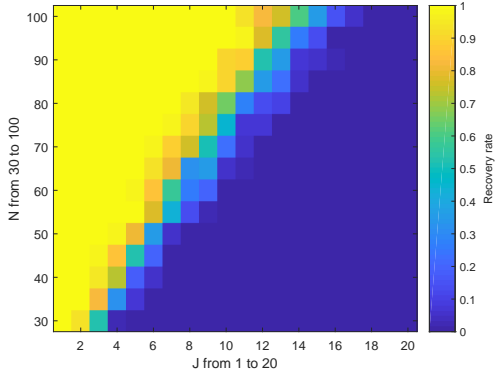}
\end{center}
   \caption{The nearly linear relation between the dimension of the observed signal, $N$, and the number of committed atoms, $J$, in terms of the success recovery rate when $\A$ is a random Gaussian matrix. }
\label{fig:KJsepGaussianJ}
\end{figure}

\begin{figure}[t]
\begin{center}
   \includegraphics[width=0.8\linewidth]{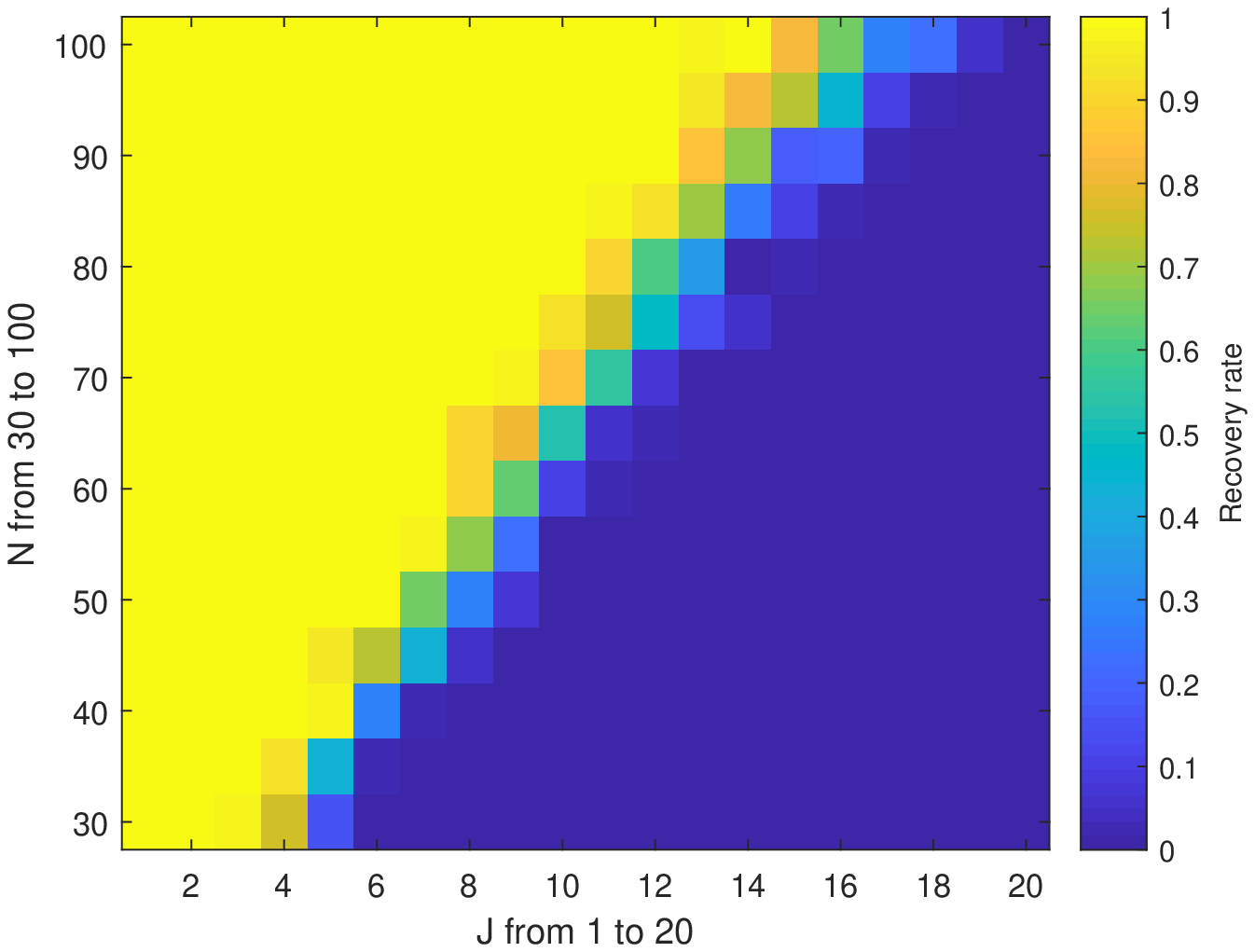}
\end{center}
   \caption{The nearly linear relation between the dimension of the observed signal, $N$, and the number of committed atoms, $J$, in terms of the success recovery rate when $\A$ is a random Fourier matrix. }
\label{fig:KJsepFourierJ}
\end{figure}

\begin{figure}[t]
\begin{center}
   \includegraphics[width=0.8\linewidth]{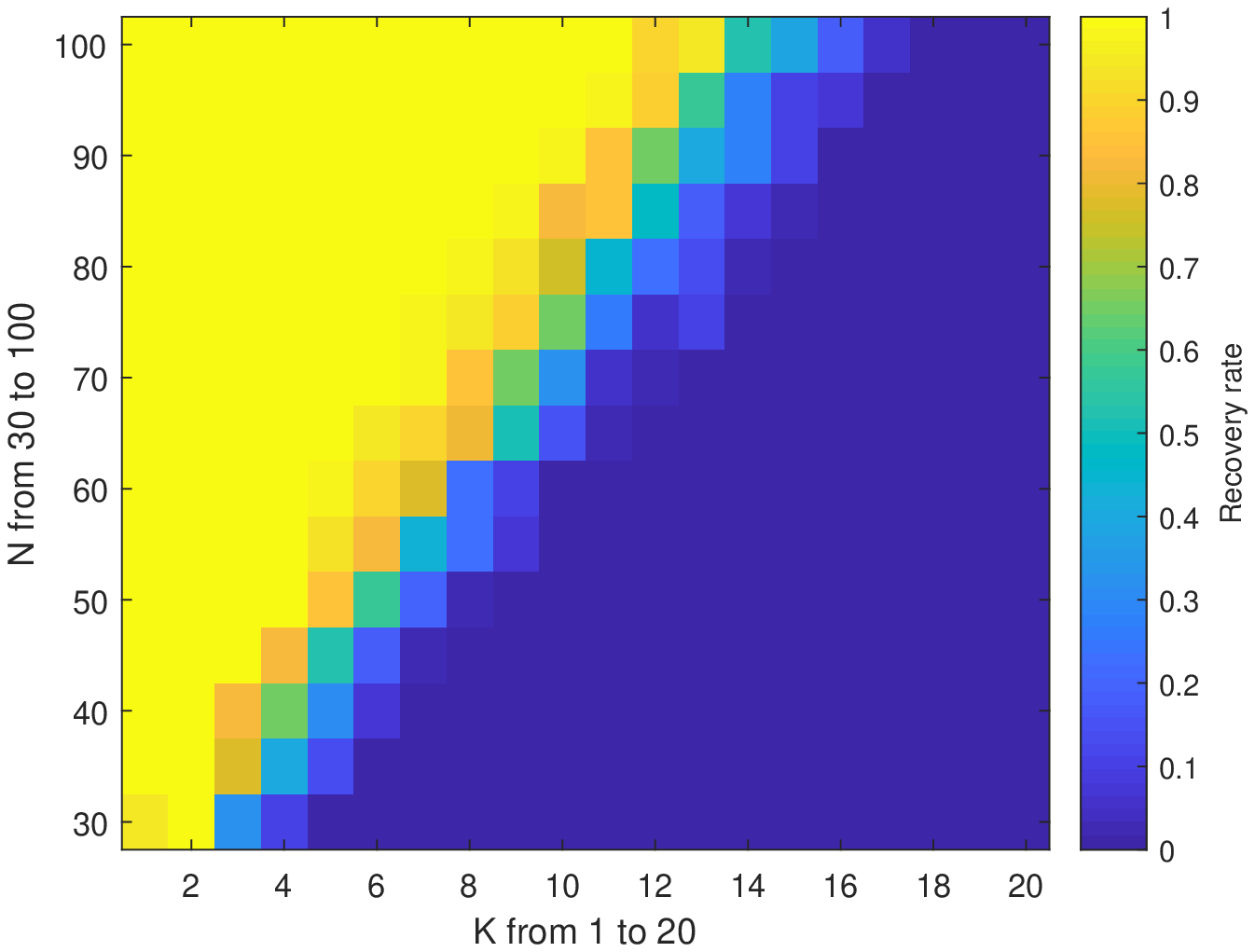}
\end{center}
   \caption{The nearly linear relation between the dimension of the observed signal, $N$, and the subspace dimension, $K$, in terms of the success recovery rate when $\A$ is a random Gaussian matrix. }
\label{fig:KJsepGaussianK}
\end{figure}

\begin{figure}[t]
\begin{center}
   \includegraphics[width=0.8\linewidth]{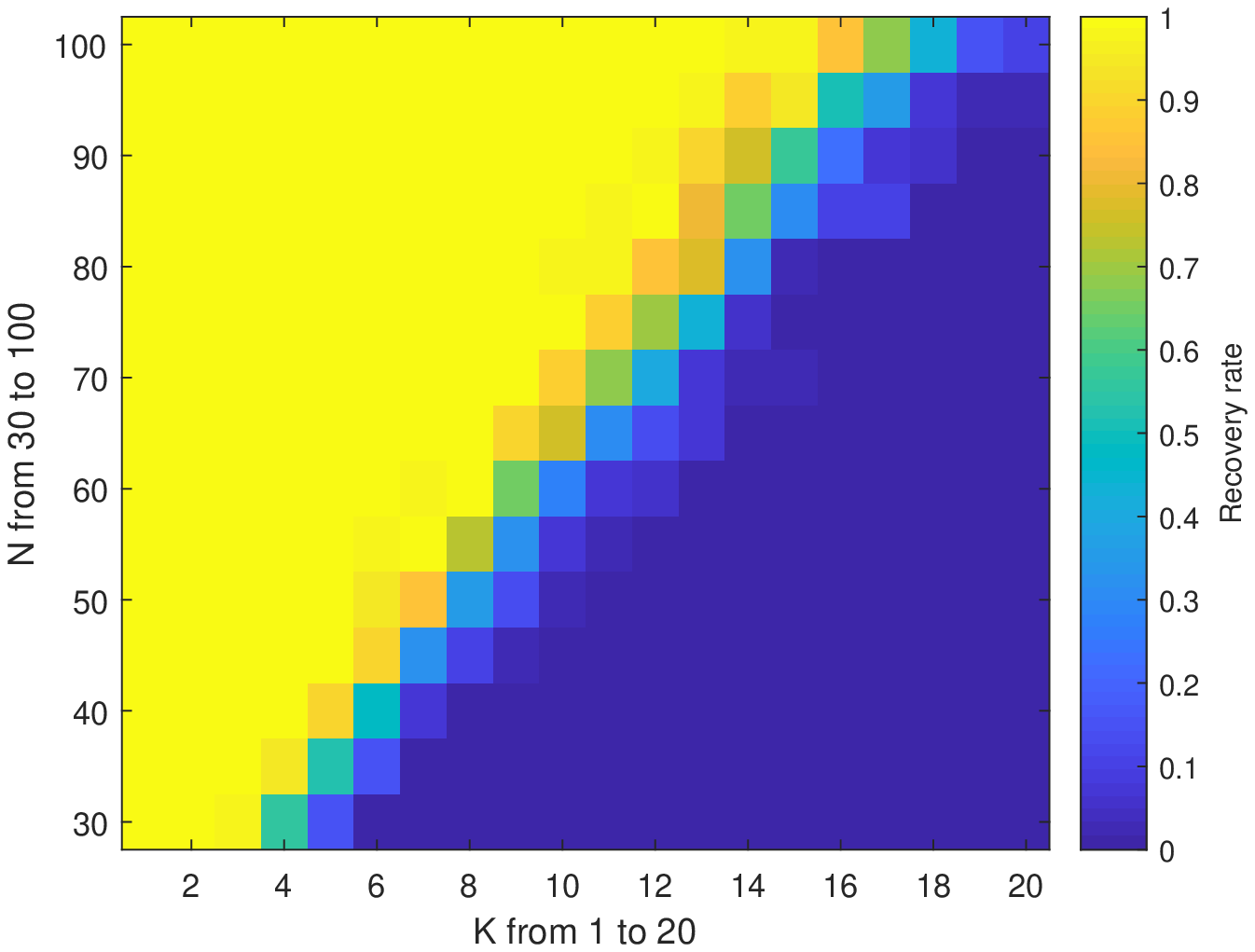}
\end{center}
   \caption{The nearly linear relation between the dimension of the observed signal, $N$, and the subspace dimension, $K$, in terms of the success recovery rate when $\A$ is a random Fourier matrix. }
\label{fig:KJsepFourierK}
\end{figure}

\subsection{The Recovery Error Bound with Noisy Measurement}

To test the noisy case, we set $M=200$, $K=J=5$, and $N=100$, and we let $\y=\L(\X_0)+\n$ with $||\n||_2\leq\eta$. Theorem \ref{maintheom2} gives a recovery guarantee of the form $||\hat{\X}-\X_0||_F\leq C\cdot \eta$ for a constant $C$ . Therefore, after dividing both sides by $||\X_0||_F$, setting $||\n||_2=\eta$ and changing the units to decibels (dB), we obtain
\begin{equation}
\begin{aligned}
20\log_{10}\left(\frac{||\hat{\X}-\X_0||_F}{||\X_0||_F}\right)\leq 20\log_{10}&\left(\frac{||\n||_2}{||\X_0||_F}\right)+\\
& \qquad\qquad 20\log_{10}(C).
\label{noisysim}
\end{aligned}
\end{equation}
We call $20\log_{10}\left(\frac{||\hat{\X}-\X_0||_F}{||\X_0||_F}\right)$ the relative error in dB and $20\log_{10}\left(\frac{||\n||_2}{||\X_0||_F}\right)$ the noise-to-signal ratio in dB. To examine the linear relation between the relative error and the noise-to-signal ratio in equation (\ref{noisysim}), we sample the real and complex components of the noise vector $\n$ independently from a standard Gaussian distribution and scale $||\n||_2$ to attain different noise-to-signal ratios. Similar to the previous plots, 40 independent simulations are run for each noise-to-signal ratio and the range of the standard deviation and mean (computed before transforming to dB) of the relative error in dB are recorded in Figs. \ref{fig:RelativeErrorGaussian} and \ref{fig:RelativeErrorFourier}. The dashed lines show the theoretical error bound from Theorem \ref{maintheom2} which are drawn by substituting the constants derived in Section \ref{noisyC1} and \ref{noisyC2} and the system parameters into equations (\ref{gnoisy}) and (\ref{fnoisy}). The slope of each dashed line are 1. We observe that when noise-to-signal ratio is smaller than 0 dB, the relative error scales linearly with respect to the noise-to-signal ratio with slope 1 for both random Gaussian and Fourier dictionaries. This confirms that $||\hat{\X}-\X_0||_F$ grows linearly with respect to $\eta$ in Theorem \ref{maintheom2}. Moreover, if the noise dominates the observed signal, solving the problem \eqref{noisy} results in $\hat{\X}=0$ and the relative error becomes 0 dB.

\begin{figure}[t]
\begin{center}
   \includegraphics[width=0.8\linewidth]{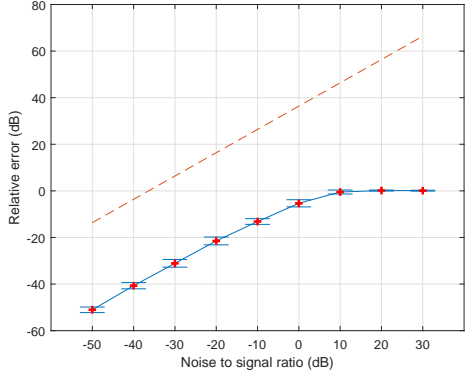}
\end{center}
   \caption{The relation between the relative error (dB) and noise-to-signal ratio (dB) when $\A$ is a random Gaussian matrix. The blue horizontal sticks and red plus sign indicate the range of the standard deviation and the mean of the relative error (dB) respectively given a specific noise-to-signal ratio (dB). The dashed line is the theoretical error bound from Theorem \ref{maintheom2}. }
\label{fig:RelativeErrorGaussian}
\end{figure}

\begin{figure}[t]
\begin{center}
   \includegraphics[width=0.8\linewidth]{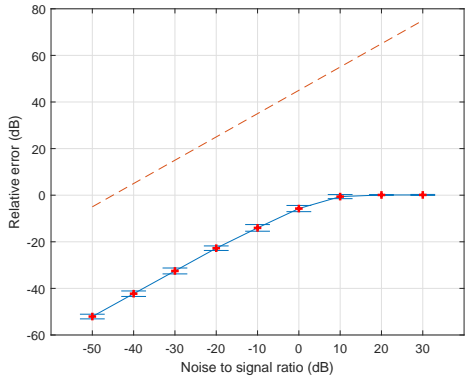}
\end{center}
   \caption{The relation between the relative error (dB) and noise-to-signal ratio (dB) when $\A$ is a random Fourier matrix. The blue horizontal sticks and red plus sign indicate the range of the standard deviation and the mean of the relative error (dB) respectively given a specific noise-to-signal ratio (dB). The dashed line is the theoretical error bound from Theorem \ref{maintheom2}.}
\label{fig:RelativeErrorFourier}
\end{figure}

\subsection{Direction of Arrival Estimation}

In this section, we apply the proposed signal model to the direction of arrival estimation problem introduced in Section \ref{DOA}. Note that there exits thousands of different subspaces that the complex calibration could live in. To give a concrete example and compare to the related work, we adopt the setting from \cite{ling2015self} where the calibration subspace $\B\in\C^{N\times K}$ is modeled by the first $K$ columns of the normalized DFT matrix $\frac{1}{\sqrt{N}}\F\in \C^{N\times N}$. The entries of $\h_j$ are sampled independently from the standard normal distribution and $\h_j$ is normalized to have unit norm. Moreover, we set $M=181$ and discretize the direction of arrival into $\theta_j=\{0,1,\cdots,180\}$ degrees. When the distance between array elements is half of the operating wavelength, we can obtain $\A$ by substituting $d = \frac{\lambda}{2}$ and $\theta_j$ into $\a(\theta_j)$ defined in Section \ref{DOA}. Furthermore, we set $N=50$ and $K=J=5$ where the directions of arrival of the 5 sources are $\{67,75, 92, 127,133\}$ degrees and the signal magnitudes are sampled independently from the uniform distribution on $[0,1]$. The real and imaginary parts of the noise vector are independent random Gaussian vectors with $0$ mean and identity covariance matrix. $SNR=30$ dB. By solving the $\ell_{2,1}$ norm minimization problem in \eqref{noisy}, the index of the nonzero column in the solution $\hat{\X}$ indicates the direction of arrival and the norm of the nonzero column indicates the signal strength. The result is recorded in Fig .\ref{fig:DOA} (a). As a comparison, we also apply the Sparselift method proposed in \cite{ling2015self} to this problem, which assumes $\D_j$ for all $j$ are the same and solves an $\ell_1$ norm minimization problem. The result is recorded in Fig .\ref{fig:DOA} (b).

\begin{figure}[h]
   \centering
   \subfloat[][The proposed method.]{\includegraphics[width=.4\textwidth]{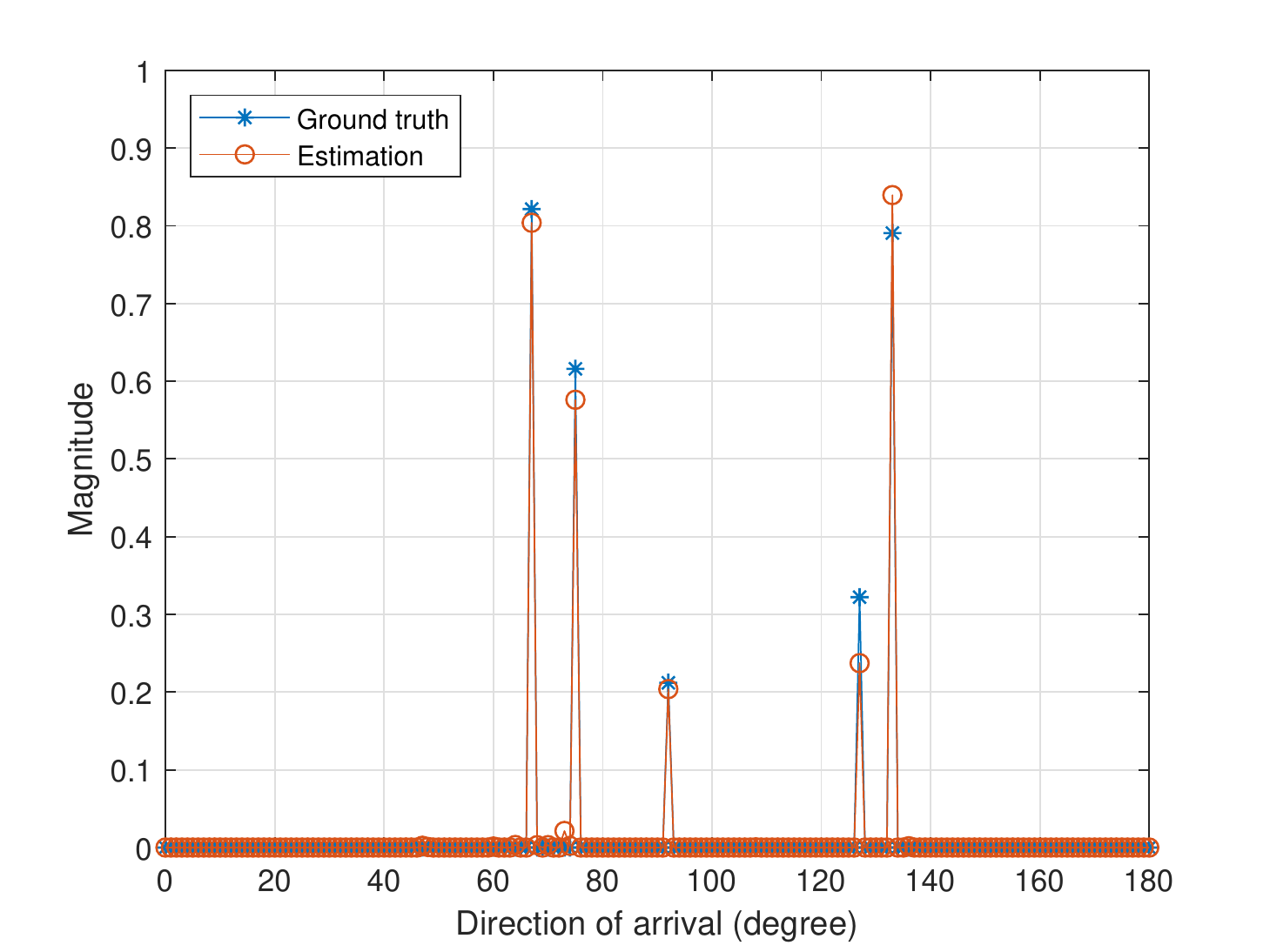}}
\\
      \subfloat[][The Sparselift method using $\ell_1$ minimization \cite{ling2015self}.]{\includegraphics[width=.4\textwidth]{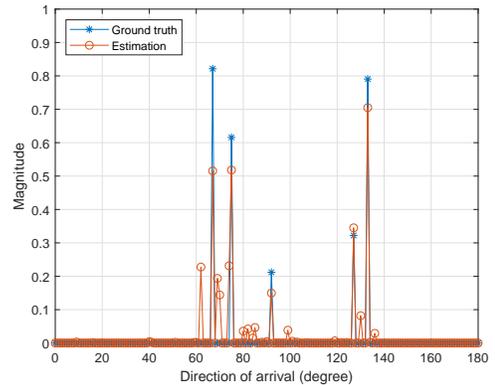}}
   \caption{The direction of arrival (DOA) estimation. (a) The estimated directions of arrival by solving the $\ell_{2,1}$ norm minimization in \eqref{noisy}. (b) The result by applying the Sparselift method using $\ell_1$ minimization proposed in \cite{ling2015self}.  }
   \label{fig:DOA}
\end{figure}

\subsection{Single Molecule Imaging}

Furthermore, we apply the proposed signal model to the single molecule imaging described in Section \ref{smi}. All data comes from the Single-Molecule Localization Microscopy grand challenge organized by ISBI \cite{ISBI2013} which contains 12,000 low-resolution frames. Each low-resolution frame is 64 pixel $\times$ 64 pixel with pixel size 100 nm $\times$ 100 nm, so that $N=64\times64=4096$. A typical, observed frame is shown in Fig. \ref{fig:PSF} (a). Superimposing all the observed frames leads to the low-resolution structure in Fig. \ref{fig:PSF} (b). The target of this experiment is to recover the high resolution image of size 320 pixel $\times$ 320 pixel, which implies that $M=320\times320=102400$, whose pixel is of size 20 nm $\times$ 20 nm. In addition, according to the statistic of the dataset, the number of activated fluorophores in each frame is less or equal to $J=17$ and we use the Gaussian point spread functions to approximate the point spread functions of the microscope. By implementing the SVD on the Gaussian point spread functions with different variances, we obtain a $K=3$ dimension subspace that point spread functions live in. Then by solving an $\ell_{2,1}$ norm regularized least square minimization problem on each low-resolution frame, we get totally 12,000 high resolution images and superimposing all the high resolution images results in the super-resolution output recorded in Fig. \ref{fig:PSF} (c).

\begin{figure}[h]
   \centering
   \subfloat[][An observed frame.]{\includegraphics[width=.21\textwidth]{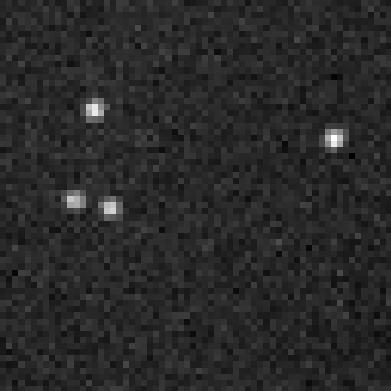}}
\\
      \subfloat[][The low-resolution structure.]{\includegraphics[width=.21\textwidth]{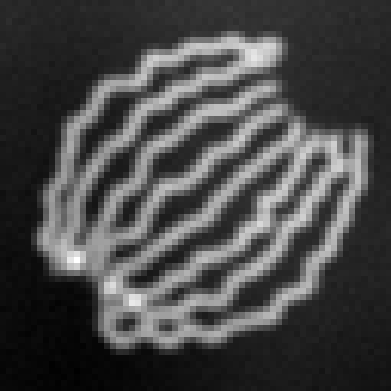}}
   \quad
   \subfloat[][The super-resolution output.]{\includegraphics[width=.21\textwidth]{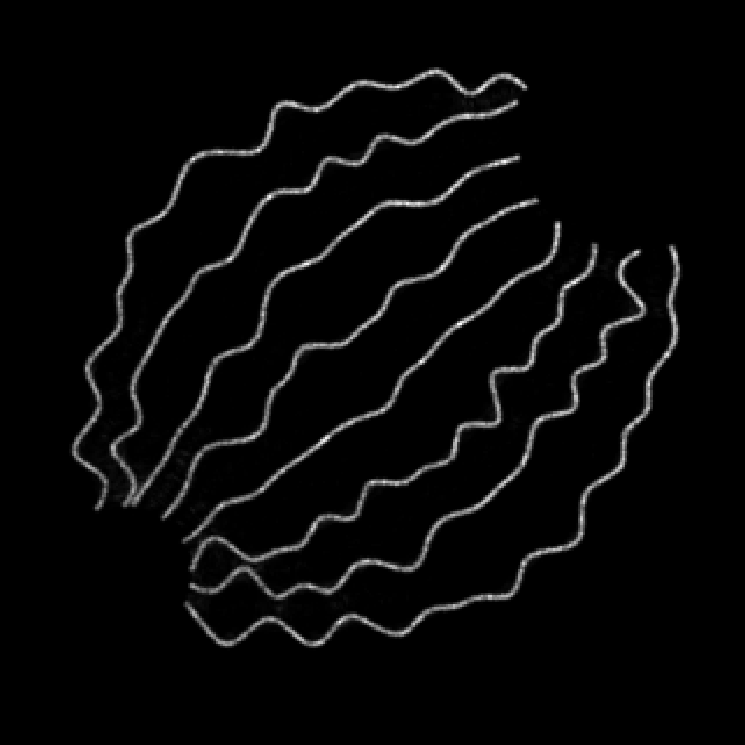}}
   \caption{The single molecule imaging. (a) The size of observed frame is 64 pixel$\times$ 64 pixel and each pixel is of size 100 nm$\times$ 100 nm. (b) Superposition of all observed frames. (d) Superposition of all recovered super-resolution images. The recovered image is of size 320 pixel$\times$ 320 pixel with pixel size 20 nm$\times$ 20 nm.}
   \label{fig:PSF}
\end{figure}

\section{Conclusion}
\label{sec:conclusion}

In this paper, we introduce the generalized sparse recovery and blind demodulation model and achieve sparse recovery and blind demodulation simultaneously.
Under the assumption that the modulating waveforms live in a known common subspace, we employ the lifting technique and recast this problem as the recovery of a column-wise sparse matrix from structured linear measurements. In this framework, we accomplish sparse recovery and blind demodulation simultaneously by minimizing the induced atomic norm, which in this problem corresponds to $\ell_{2,1}$ norm minimization. In the noiseless case, we derive near optimal sampling complexity that is proportional to the number of degrees of freedom, and in the noisy case we bound the recovery error of the structured matrix. Numerical simulations support our theoretical results. In addition to extending the class of dictionaries we have considered, an interesting future direction would be to relax the constraint that each $\D_j$ is diagonal while preserving the low-dimensional subspace assumption.

\section*{Acknowledgment}
This work was supported by NSF grant CCF$-1704204$. 



%
\bibliographystyle{ieeetr}
\bibliography{test}

%




\end{document}